%% file: main.tex
\documentclass[11pt]{article}
\pdfoutput=1 

\usepackage{lmodern}

\usepackage[margin=1in]{geometry}
\usepackage{amsfonts,amssymb,amsmath,amsthm,mathtools,thmtools}
\usepackage{thm-restate}
\usepackage{microtype,xspace,bm,floatrow}

\usepackage{doi}
\usepackage{hyperref}
\usepackage[dvipsnames]{xcolor}
\definecolor{newcolor}{hsb}{0.6,1,0.75}
\hypersetup{
	colorlinks=true,
	urlcolor=newcolor,
	linkcolor=newcolor,
	citecolor=newcolor,
	unicode
}

\usepackage{bold-extra}

\usepackage{enumitem}
\setlist{  
  listparindent=\parindent,
  parsep=0pt,
}

\usepackage[capitalise,nameinlink,noabbrev]{cleveref}
\crefname{equation}{}{}

\usepackage{tikz}
\usepackage{graphicx}
\usetikzlibrary{positioning,arrows.meta,math,shapes.geometric,decorations.pathmorphing,decorations.pathreplacing}

\usepackage[font=small,labelfont=bf]{caption} 
\usepackage{subcaption}
\usepackage[percent]{overpic}

\usepackage[english]{babel}
\addto\extrasenglish{}
\addto\extrasenglish{}

\usepackage{lpic}

\makeatletter
\DeclareRobustCommand\bfseries{%
  \not@math@alphabet\bfseries\mathbf
  \fontseries\bfdefault\selectfont\boldmath}
\makeatother

\theoremstyle{definition}
\declaretheorem[name=Theorem, numberwithin=section]{theorem}

\declaretheorem[name=Lemma,sibling=theorem]{lemma}
\declaretheorem[name=Corollary, sibling=theorem]{corollary}

\declaretheorem[name=Definition, sibling=theorem, style=definition]{definition}

\declaretheorem[name=Conjecture]{conj}

\newcommand{\poly}{\textrm{\upshape poly}}
\newcommand{\negl}{\textrm{\upshape negl}}

\newcommand{\newclass}[2]{\newcommand{#1}{{\text{\upshape\sffamily #2}}\xspace}}
\renewcommand{\P}{{\text{\upshape\sffamily P}}\xspace}
\newclass{\NP}{NP}
\newclass{\coNP}{coNP}
\newclass{\FP}{FP}
\newclass{\TFNP}{TFNP}
\newclass{\TFAP}{TFAP}
\newclass{\PLS}{PLS}
\newclass{\PPA}{PPA}
\newclass{\PPAD}{PPAD}
\newclass{\PPADS}{PPADS}
\newclass{\PPP}{PPP}
\newclass{\PAP}{PAP}
\newclass{\SAP}{SAP}
\newclass{\PWPP}{PWPP}
\newclass{\CLS}{CLS}
\newclass{\EOPL}{EOPL}
\newclass{\SOPL}{SOPL}
\newclass{\UEOPL}{UEOPL}
\newclass{\RAMSEY}{RAMSEY}
\newclass{\BIRAM}{BIRAMSEY}
\newclass{\cA}{A}
\newclass{\cB}{B}
\newclass{\BPP}{BPP}
\newclass{\PLC}{PLC}
\newclass{\UPLC}{UPLC}
\newclass{\PiH}{PiH}
\newclass{\BQP}{BQP}

\newcommand{\newprob}[2]{\newcommand{#1}{{\text{\upshape\scshape #2}}\xspace}}
\newprob{\leaf}{Leaf}
\newprob{\eol}{EoL}
\newprob{\eolLong}{End-of-Line}
\newprob{\sol}{SoL}
\newprob{\solLong}{Sink-of-Line}
\newprob{\iter}{Iter}
\newprob{\sod}{SoD}
\newprob{\sodLong}{Sink-of-Dag}
\newprob{\kkt}{KKT}
\newprob{\pA}{A}
\newprob{\pB}{B}
\newprob{\eopl}{EoPL}
\newprob{\eoplLong}{End-of-Potential-Line}
\newprob{\ueoplLong}{Unique-EoPL}
\newprob{\ueopl}{UEoPL}
\newprob{\eoml}{EoML}
\newprob{\eomlLong}{End-of-Metered-Line}
\newprob{\sopl}{SoPL}
\newprob{\soplLong}{Sink-of-Potential-Line}
\newprob{\pigeon}{Pigeon}
\newprob{\wpigeon}{WeakPigeon}
\newprob{\lonely}{Lonely}
\newprob{\reversiblepigeon}{RPigeon}
\newprob{\reversiblepigeonlong}{Reversible-Pigeon}
\newprob{\factoring}{Factoring}
\newprob{\nash}{Nash}
\newprob{\Or}{Or}
\newprob{\ramsey}{Ramsey}
\newprob{\biramsey}{BiRamsey}
\newprob{\sunflower}{Naive Sunflower}
\newprob{\uplc}{UPLC}
\newprob{\tuplc}{T-UPLC}

\newcommand{\cnf}{\ensuremath{\neg\mathrm{Total}}}

\newcommand{\set}[1]{\ensuremath{\{#1\}}}
\renewcommand{\restriction}{\!\upharpoonright\!}

\newcommand{\RR}{\ensuremath{\mathbb{R}}}

\newcommand{\ZZ}{\ensuremath{\mathbb{Z}}}

\newcommand{\Xcomment}[1]{{}}

\newcommand{\AZC}{{\sc AllZeroColumn}}

\newcommand{\E}{\mathbb{E}}
\newcommand{\psE}{\widetilde\E}

\newcommand{\B}{\{0,1\}}
\newcommand{\ourPigeon}[3]{%
  #1\textup{\textsc{-Pigeon}}^{#2}_{#3}%
}
\newcommand{\tightPigeon}[1]{%
  #1\textup{\textsc{-Pigeon}}%
}
\newcommand{\pmap}{h}

\newcommand{\ps}{p}
\newcommand{\uEvent}{\mathsf{NonWit}}
\newcommand{\hitI}{\mathsf{Hit\textrm{-}i}}
\newcommand{\hitJ}{\mathsf{Hit\textrm{-}j}}

\setlist[description]{leftmargin=\parindent,labelindent=\parindent}

\begin{document}

\mbox{}\vspace{12mm}

\begin{center}
{\huge On Pigeonhole Principles and Ramsey in  $\TFNP$}
\\[1cm] \large
	
\setlength\tabcolsep{2em}
\begin{tabular}{cccc}
Siddhartha Jain&
Jiawei Li&
Robert Robere&
Zhiyang Xun\\[-1mm]
\small\slshape UT Austin &
\small\slshape UT Austin & 
\small\slshape McGill &
\small\slshape UT Austin
\end{tabular}
	
\vspace{6mm}
	
\large
\today

\vspace{6mm}
	
\end{center}

\begin{abstract}
\noindent


We show that the $\TFNP$ problem \ramsey is not black-box reducible to $\pigeon$, refuting a conjecture of Goldberg and Papadimitriou in the black-box setting. We prove this by giving reductions to \ramsey from a new family of \TFNP problems that correspond to generalized versions of the pigeonhole principle, and then proving that these generalized versions cannot be reduced to $\pigeon$. Formally, we define $t$-$\PPP$ as the class of total $\NP$-search problems reducible to finding a $t$-collision in a mapping from $(t-1)N+1$ pigeons to $N$ holes. These classes are closely related to multi-collision resistant hash functions in cryptography. We show that the generalized pigeonhole classes form a hierarchy as $t$ increases, and also give a natural condition on the parameters $t_1, t_2$ that captures exactly when $t_1$-$\PPP$ and $t_2$-$\PPP$ collapse in the black-box setting. Finally, we prove other inclusion and separation results between these generalized $\pigeon$ problems and other previously studied $\TFNP$ subclasses, such as $\PLS, \PPA,$ and $\PLC$. Our separation results rely on new lower bounds in propositional proof complexity based on pseudoexpectation operators, which may be of independent interest.
\end{abstract}


\section{Introduction}
The theory of \TFNP is the study of \NP search problems that are \emph{guaranteed} to have solutions. In most problems studied in the literature, this guarantee is usually due to a non-constructive combinatorial lemma \cite{JPY88,Pap94}. Perhaps the most famous example is the Pigeonhole Principle (PHP), which defines the search problem \pigeon: given a polynomial-size circuit encoding a mapping from $N+1$ pigeons to $N$ holes, find two pigeons that are in the same hole. However, many other non-constructive combinatorial principles also play important roles in \TFNP, including variants of the Handshaking Lemma (corresponding to the classes $\PPA$, $\PPAD$) \cite{Pap94, Beame95}, the guaranteed convergence of local search algorithms (corresponding to the class $\PLS$) \cite{JPY88, Pap94}, and the guaranteed convergence of gradient descent (corresponding to the class $\CLS$~\cite{DP11CLS, FGHS23}).
Importantly, \emph{polynomial-time reducibility} between these various search problems corresponds directly to how \emph{relatively constructive} the corresponding combinatorial principles are. 
From this viewpoint, \TFNP can be seen as a kind of ``bounded reverse mathematics'', in which we seek to study what families of non-constructive combinatorial principles can have their witnesses \emph{constructively} reduced to each other.

One problem which has so far resisted classification is the search problem corresponding to \emph{Ramsey's Theorem}. 
In the \ramsey problem, introduced by Krajíček \cite{Kra05}, we are given a polynomial-size circuit $C$ encoding the edge relation of a graph on $N = 2^{n+1}$ vertices, and our goal is to find a clique or independent set of size $n$ on the graph. 
Since the underlying combinatorial principle comes from extremal combinatorics, it is natural to ask if there is a polynomial-time reduction to or from the Pigeonhole Principle, the prototypical example of extremal reasoning.
This seems especially reasonable given that standard proof of Ramsey's theorem is essentially a \emph{recursive} application of the pigeonhole principle\footnote{Start by picking a vertex $v$, and then delete all vertices that are adjacent to $v$ or non-adjacent to $v$, whichever set is larger, and then repeat this process. Continuing in this way we can construct a sequence of $\log n$ vertices, and at least half of these must be either all adjacent or non-adjacent to each other.}.
Whether or not \ramsey is in \PPP --- the \TFNP subclass defined by polynomial-time reductions \pigeon --- has been informally asked many times, was formally conjectured by Goldberg and Papadimitriou \cite{GP17}, and also appeared in Daskalakis's recent ICM plenary lecture~\cite[Open Question 17]{daskalakis2019equilibria}.

\begin{conj}[Goldberg \& Papadimitriou \cite{GP17}]
    \label{conj:ramsey-ppp}
    \ramsey is in \PPP.
\end{conj}

However, despite the attention given to this problem, actually finding a reduction from $\ramsey$ to $\PPP$ has remained elusive.
Recently, a line of work has weakened this goal, and instead attempted to place \ramsey into \emph{any} natural subclass of \TFNP \cite{Thapen22, Pasarkar2023,Bourneuf23}. 
Of particular note is the work of Pasarkar, Papadimitriou, and Yannakakis~\cite{Pasarkar2023}, who defined a novel new \TFNP subclass called \PLC.
This class exactly captures the kinds of ``recursive applications'' of the pigeonhole principle that are used in the typical proof of Ramsey's Theorem and other extremal combinatorial principles, like the Erd\"{o}s-Rado Sunflower Lemma. 
Pasarkar et al.~showed that both \ramsey and \pigeon were contained in $\PLC$, but left open the question of determining how much extra power $\PLC$ contains over $\PPP$.



\subsection{Our Results}

Our main result is to negatively resolve \cref{conj:ramsey-ppp} in the \emph{black-box setting}, where the inputs are provided by oracle queries.
\begin{theorem}
    \label{thm:ramsey-pigeon}
    There is no black-box reduction from \ramsey to \pigeon.
\end{theorem}

We note that \emph{all} known upper-bound techniques in \TFNP also work in the black-box setting, so our theorem rules out any approach to the above conjecture using current technology. Any unconditional white-box separation within \TFNP would imply $\P \neq \NP$, thus, one can only hope to prove separations in the black-box setting without a breakthrough in complexity theory.
Also, as an immediate corollary of our previous theorem, we conclude that $\PLC$ strictly contains $\PPP$ with respect to black-box reductions, answering a question of Pasarkar, Papadimitriou, and Yannakakis \cite{Pasarkar2023}.

To prove \cref{thm:ramsey-pigeon}, we introduce a new family of problems in $\TFNP$ that correspond to \emph{generalized} Pigeonhole Principles and systematically study their properties. Whereas the Pigeonhole Principle says that any method of placing $n+1$ pigeons into $n$ holes must result in a collision of \emph{two} pigeons in one hole, the generalized Pigeonhole Principle we study says that any method of placing $t(n+1)$ pigeons into $n$ holes must result in a collision of $t+1$ pigeons in a hole. We can encode this as a total search problem as follows.
\begin{restatable}{definition}{ourpigeon}
\label{def:pigeon}
    For any positive integer $n$ let $t(n), M(n)$, and $N(n)$ be integer parameters satisfying $M > (t-1)N$. The $\ourPigeon{t}{M}{N}$ problem is defined as follows.
    
    \begin{description}
        \item[Input] $(n, h)$, where $n$ is given in unary and $h:[M] \rightarrow [N]$ is a map of $M$ pigeons to $N$ holes represented by a $\poly(n)$-size circuit.
        \item[Solutions] A \emph{$t$-collision} in $h$, which is a set of $t$ pigeons that are all assigned to the same hole by the circuit.
    \end{description}
\end{restatable}
We note that for this problem to be in \TFNP, $t(n)$ must have at most polynomial growth rate, as otherwise a solution will be too large to verify in polynomial time. 
As an important special case, we use $\tightPigeon{t}$ to indicate $\ourPigeon{t}{M}{N}$ with the tight parameter setting $M = (t-1)N+1$ --- note that $\tightPigeon{2} = \pigeon$, for instance.
Finally, $N = 2^n$ in the typical setting of interest, and so we will use ``$n$'' and ``$\log N$'' interchangeably as parameters.

The $\pigeon$ problem and its $\tightPigeon{t}$ variants are naturally related to cryptography, as solving them efficiently would imply the ability to find collisions in collision-resistant hash functions. These hash functions usually compress the output from say $2n$ to $n$ bits, hence ore relevant here is the `weak' version of $\pigeon$, which defines the class $\PWPP$.
While they have been implicitly studied before by cryptographers (see Related Work in \cref{sec:concurrent} for more details), they are studied systematically\footnote{Also see the recent independent work of \cite{BGS24merrier}, in which these problems are also studied.} here from the perspective of $\TFNP$ for the first time.
To relate these problems to $\ramsey$, we generalize an argument due to Komargodski, Naor, and Yogev \cite{KMY18} to obtain the following:

\begin{restatable}{theorem}{ramseyinstantiated}
    \label{thm:ramseyinstantiated}

    Whenever $M \ge N^{4t}/4^t$, $\ourPigeon{t}{M}{N}$ can be black-box reduced to $\ramsey$ and its bipartite variant $\biramsey$.
\end{restatable}

We note that Krajíček \cite{Kra05} showed that $\ourPigeon{2}{2N}{N}$ reduces to $\ramsey$, which is generalized by the above theorem.
The previous theorem implies, for instance, that there is a black-box reduction from $\ourPigeon{t}{N^{4t}}{N}$ to $\ramsey$, and so to prove \cref{thm:ramsey-pigeon} we will argue that there is \emph{no} black-box reduction from $\ourPigeon{t}{N^{4t}}{N}$ to $\pigeon$.
In order to do this we develop and apply techniques from \emph{propositional proof complexity}, which are described in more detail in \cref{sec:our-techniques} and may be of independent interest.
For now, we note that $\ourPigeon{t}{N^{4t}}{N}$ is, at first glance, seemingly incomparable with $\pigeon$ in power.
On one hand, $\ourPigeon{t}{N^{4t}}{N}$ is much \emph{weaker} than $\pigeon$, as the number of pigeons is so much greater than the number of holes.
On the other hand, in this problem we are seeking a collision of $t$ pigeons in a hole, and it is not at all clear how to do this when we can only guarantee a collision of \emph{two} pigeons.

\subsubsection{The Pecking Order}

Indeed, we will prove much more than the non-reducibility of $\ourPigeon{t}{N^{4t}}{N}$ to $\pigeon$.
In fact, we exhibit a \emph{hierarchy theorem}: $\ourPigeon{t}{M}{N}$ does not black-box reduce to $\ourPigeon{(t-1)}{M'}{N'}$, \emph{no matter the compression rates of pigeons to holes in either problem}.
Before we formally state our next results, let us first introduce complexity classes capturing these generalized pigeon problems.

\begin{definition}[$t$-\PPP and $t$-\PWPP]\label{def:t_ppp}
    For any function $t(n) \geq 2$, define the classes
    \begin{description}
        \item[{$t$-\PPP}.] All search problems reducible to $\tightPigeon{t}$.
        \item[{$t$-\PWPP}.] All search problems reducible to $\ourPigeon{t}{M}{N}$, with $M = (t-1+c)N$ for constant $c > 0$.
    \end{description}
\end{definition}


Note that we can define the classic \TFNP classes \PPP and \PWPP by taking $t=2$ in \cref{def:t_ppp}. Trivially, we have $t$-\PWPP $\subseteq$ $t$-\PPP for any $t \geq 2$. Furthermore, a simple padding argument (cf. \cref{lem:hierarchy}) shows that these classes form a hierarchy: $t$-\PPP is contained in $(t+1)$-\PPP for each $t$.
We call the hierarchy of classes for \emph{constant} values of $t$ the \emph{Pigeon Hierarchy}, denoted by \PiH.
\begin{definition}[Pigeon Hierarchy]
    $\PiH = \displaystyle \bigcup_{t=2}^{\infty} t\text{-}\PPP$
\end{definition}

Lying atop the Pigeon Hierarchy are the $t$-\PPP classes where $t(n)$ is a \emph{growing} function of the input size.
We isolate two interesting cases here.
The strongest we call \PAP, for \emph{Polynomial Averaging Principle}.
\begin{definition}
    \PAP is the set of all search problems reducible to $n\textrm{-}\PPP$.
\end{definition}

By a fairly simple argument (cf.~\cref{thm:polynomial-robustness}), one can show that $t(n)$-\PPP is contained in $\PAP$ for \emph{any} polynomial function $t(n)$.
Since $t(n)$ must have polynomial growth rate in order for $\tightPigeon{t}$ to be in \TFNP, this means that $\PAP$ is the strongest possible class of generalized pigeon problems.
Similar to this definition, we introduce the class \SAP (for \emph{Subpolynomial Averaging Principle}).

\begin{definition}
   \SAP is the set of all search problems reducible to $t(n)\text{-}\PPP$, for some $t(n)$ sub-polynomial\footnote{A function $t(n)$ is said to be sub-polynomial if $t(n) = o(n^c)$ for any $c>0$.} in $n$.
\end{definition}

Intuitively, $\SAP$ contains all problems defined by $\tightPigeon{t}$, but excludes the hardest problems in $\PAP$, which is convenient when stating our results in the strongest form. We call this entire collection of complexity classes the \emph{Pecking Order}.

\subsubsection{Structure of the Pecking Order}

We are able to \emph{completely} characterize the relationships between various $t$-\PPP classes in the black-box setting.
In the black-box setting, instead of providing the inputs succinctly via polynomial-size boolean circuits, we instead think of the inputs as being provided as black-box oracles.
For example, in the black-box version of the $\pigeon$ problem, the map from pigeons to holes is provided by an oracle $f$ which, given the name of the pigeon, outputs the hole that the pigeon maps to.
In general, if $\pA$ is a total search problem defining a $\TFNP$ subclass $\cA$ via black-box reducibility, we will use $\pA^{dt}$ ($\cA^{dt}$, resp.) to denote the black-box versions of this problem and class (we refer to \cref{sec:prelim} for formal definitions).

To start, on the side of separations, we are able to prove that the Pigeon Hierarchy (\PiH) is strict. 
Indeed, we can prove very strong black-box separations between $\ourPigeon{t}{M}{N}$ and $\ourPigeon{(t+1)}{M'}{N'}$. 

\begin{restatable}{theorem}{pihsep}
\label{cor:pih-sep}
    If $t$ is a constant and $M, N, M', N'$ are parameters chosen so that $M' = \poly(M)$, $M \geq tN + 1$, $M' \geq (t-1)N' + 1$,
    then $\ourPigeon{(t+1)}{M}{N}$ does not have an efficient black-box reduction to $\ourPigeon{t}{M'}{N'}$.
    In particular, $(t+1)\text{-}\PPP^{dt} \not \subseteq t\text{-}\PPP^{dt}$ for any constant $t$, and so $\PiH^{dt}$ forms a strict hierarchy in the black-box setting.
\end{restatable}


Note that the above separation is invariant of compression rate, and so it even shows the stronger separation $(t+1)$-$\PWPP^{dt} \not \subseteq t$-$\PPP^{dt}$.
We can also prove separation results for non-constant (growing) $t$. 
Before stating this generalization, let us introduce a helpful definition for relating the various collision rates.
\begin{restatable}{definition}{PolynomialClose}\label{def:polynomially close}
    A function $b(n)$ is \emph{polynomially close} to $a(n)$ if there exists a polynomial $p(n)$ such that for any $n$, $b(n) \leq a(p(n))$, and $a(n) \leq b(p(n))$.
\end{restatable}

For example, all functions with polynomial growth rate are polynomially close to each other, while slower-growing functions like poly-logarithms are not.
This definition is useful since it captures the $t$-\PPP classes in the following sense:

\begin{restatable}{theorem}{PolynomialRobustness}
    \label{thm:polynomial-robustness}
    If a function $b(n)$ is polynomially close to $a(n)$, then $a\textrm{-}\PPP = b\textrm{-}\PPP$.
\end{restatable}

The proof of \cref{thm:polynomial-robustness} uses a fairly standard \emph{copying argument} to manipulate the parameters, and we leave the details to \cref{sec:prelim:PO}.
By generalizing our separation result for the Pigeon Hierarchy (\cref{cor:pih-sep}), we can prove a converse to \cref{thm:polynomial-robustness}.

\begin{restatable}{theorem}{generalLB}\label{thm:general_lb}
    Let $a(n), b(n)$ be functions such that $a(n)$ is larger than $b(n)$ and is not polynomially close to $b(n)$. Let $M, N, M', N'$ be any parameters chosen so that $M' = \poly(M)$, $M \geq (a(n)-1)N + 1$, $M' \geq (b(n)-1)N' + 1$. Then $\ourPigeon{a(n)}{M}{N}$ does not have an efficient black-box reduction to $\ourPigeon{b(n)}{M'}{N'}$.
\end{restatable}

Since the constant function $a(n)=(t+1)$ is not polynomially close to $b(n) = t$ for any constant $t$, this is a strict generalization of \cref{cor:pih-sep}.
Combining \cref{thm:polynomial-robustness} with \cref{thm:general_lb}, we can completely characterize the structure of the Pecking Order with respect to the collision number.

\begin{restatable}{theorem}{posep}
    \label{thm:posep}
    For any two functions $a(n), b(n)$, $a(n)$-$\PPP^{dt} = b(n)$-$\PPP^{dt}$ if and only if $a(n)$ and $b(n)$ are polynomially close.
\end{restatable}

We emphasize that \cref{cor:pih-sep} and \cref{thm:general_lb} do not depend on the specific choice of parameters $M, N, M', N'$, as long as they are non-trivial.    
For example, $\ourPigeon{3}{N^2}{N}$ cannot be reduced to $\pigeon^{N+1}_{N}$, though the latter one has a much smaller compression rate.
Conceptually:
\begin{center}
\emph{We cannot trade a lower compression rate in exchange for more collisions}. 
\end{center}
This favorable property of our structural theorem makes it very convenient for showing separation results for problems from the Pecking Order. Once we can reduce $\ourPigeon{t}{M}{N}$ to a search problem of interest (like $\ramsey$) for \emph{any} value of $M \gg t \cdot N$, we \emph{automatically} separate the search problem from any lower level of the Pecking Order than $t$-\PPP, including \PPP. 



Conversely, we show that one cannot perform a tradeoff in the \emph{other} direction (even in the randomized setting).
Formally speaking, we separate $\PPP$ from $n$-\PWPP for any randomized black-box reduction. 
\begin{restatable}{theorem}{pppnpwpp}\label{thm:lb:compression}
   There is no randomized black-box reduction from $\pigeon$ to $\ourPigeon{n}{M}{N}$ with $M \geq (n-1+c)N$ for a constant $c > 0$.
\end{restatable}

Complementary to \cref{thm:general_lb}, this result can be interpreted as saying: 
\begin{center}
    \emph{We cannot trade more collisions in exchange for an (extremely) low compression rate.}
\end{center}
In a concurrent work by Li~\cite{Li23tfap}, \cref{thm:lb:compression} is generalized to total search problems with many solutions on any given input, like $\PWPP$ and variants. We refer to \cref{sec:concurrent} for further related discussion.

We finally remark that our separations between the Pecking Order classes can be viewed as a first step towards a notorious problem in cryptography: determining whether there is a \emph{black-box construction} of \emph{multi-collision resistant hash-functions} from \emph{collision-resistant hash-functions} (cf.~\cref{sec:concurrent}). In our language, this (very roughly) translates to proving or disproving the existence of a \emph{randomized} black-box Turing reductions from $\ourPigeon{t}{M}{N}$ to $\pigeon$. Our main separation result can be seen as a confirmation that there is no \emph{deterministic} black-box reduction, and \cref{thm:lb:compression} is a weak (many-one) separation in the converse direction. 
However, we do not obtain a randomized separation in this paper, and leave it open for future work.





\subsubsection{Ramsey and the Pecking Order}

By employing our separations for Pecking Order classes, we give strong lower bounds for the $\ramsey$ problem in the black-box setting, as well as for its bipartite variant $\biramsey$, where we are given a bipartite graph on $(N,N)$ vertices, and have to output either a $(\log N)/2$-biclique or a $(\log N)/2$-bi-independent set. In particular, combining \cref{thm:ramseyinstantiated} with \cref{thm:general_lb} we obtain the following black-box separation, which is significantly stronger than \cref{thm:ramsey-pigeon}.


\begin{restatable}{theorem}{ramseyppp}
\label{cor:ramsey-ppp}      
$\ramsey^{dt}, \biramsey^{dt} \not \in \SAP^{dt}$.
In particular, $\ramsey^{dt}, \biramsey^{dt} \not \in \PPP^{dt}$.
\end{restatable}

In the reverse direction, we show that \biramsey fits into the Pecking Order with a slightly weaker parameter. Together with \cref{cor:ramsey-ppp}, we give an almost tight characterization of \biramsey using the Pecking Order.

\begin{restatable}{theorem}{BiRamseyinPAP}(Generalization of \protect{\cite[Theorem~3.10]{Komargodski2017}})
\label{lem:biramsey-pap}
    $\frac{n - \log n}{2}\text{-}\biramsey \in \PAP$.
\end{restatable}



We also reveal an interesting connection between the Pecking Order and the \emph{Polynomial Long Choice} class (\PLC) introduced by \cite{Pasarkar2023}. Recall that \cite{Pasarkar2023} has shown that $\ramsey \in \PLC$, therefore, by \cref{cor:ramsey-ppp}, we already separate $\PLC$ from $\PPP$ in the black-box setting, resolving one of their open questions.
However, we can get even stronger results by considering a subclass of \PLC called \UPLC~\cite{Pasarkar2023} (Unary \PLC, see \cref{def:uplc}), which is the \emph{non-adaptive} variant of \PLC. In particular, we can show that $\ourPigeon{n/2}{N}{\sqrt{N}}$ reduces to \UPLC (\cref{lem:pwpp-uplc}), and therefore we can separate $\UPLC$ from $\SAP$ (and of course, $\PPP$) in the black-box setting using our separation results.

\begin{restatable}{theorem}{uplcpap}
\label{cor:uplc-ppp}
    $\UPLC^{dt} \not\subseteq \SAP^{dt}$.
\end{restatable}

Conversely, we show in \cref{lem:uplc-pap} that $\UPLC \subseteq n\textrm{-}\PWPP$, and thus by our separation regarding to the compression rate (\cref{thm:lb:compression}), we have that $\PPP^{dt} \not \subseteq \UPLC^{dt}$, resolving another open question of Pasarkar, Papadimitriou, and Yannakakis.

\begin{restatable}{theorem}{uplcsep}
\label{cor:uplc-plc}
    $\PPP^{dt} \nsubseteq \UPLC^{dt}$. Consequently, $\PLC^{dt} \nsubseteq \UPLC^{dt}$.
\end{restatable}

We also observe that $\UPLC$ does not capture the full strength of the non-adaptive iterative pigeonhole principle. We define a new problem $\tuplc$ (\cref{def: tuplc}), which is similar to $\uplc$, but has the strongest possible parameters. We show that $\tuplc$ is indeed $\PAP$-complete (\cref{lem: tuplc pap}), therefore, the non-adaptive iterative pigeonhole principle is equivalent to the generalized pigeonhole principle.

Taken together, our results provide a nearly complete picture (see \cref{fig:tfnp} for a summary) of the relative complexities of the Pecking Order, $\PLC$, $\UPLC$, $\ramsey$, and $\biramsey$ in the black-box setting.

\begin{figure}[btph]
    \centering
    \include{tfnp}
    \caption{Complexity classes and problems defined by Generalized Pigeonhole Principles and Ramsey. An arrow $\cA\rightarrow\cB$ means $\cA\subseteq \cB$. An {\color{YellowOrange} orange} arrow $\cA~{\color{YellowOrange} \rightarrow} ~\cB$ means $\cA \subseteq \cB$ and $\cB \not\subseteq \cA$ in the black-box setting. A~dashed arrow~$\cA\dashrightarrow\cB$ means $\cA\not\subseteq \cB$ in the black-box setting. All black-box separations in this figure are contributions of this work, labelled by the corresponding theorem. \\
    We refer to the tower of $t$-\PPP classes as the Pecking Order, while the union of $t$-\PPP for all constant $t$ is referred to as the Pigeon Hierarchy (\PiH). We note that while our result $\biramsey \not\in \SAP$ (in the black-box setting) applies for the standard parameter regime, $\biramsey \in \PAP$ is for a slightly smaller parameter.}
    \label{fig:tfnp}
\end{figure}

\subsubsection{$\PLS$, $\PPA$, and the Pecking Order}

Given these results, it is also quite natural to try and relate the Pecking Order classes to other well-studied and important $\TFNP$ subclasses.
We focus on the relationships with the classes $\PLS$ and $\PPA$, which are the two of the three strongest $\TFNP$ subclasses out of the ``original five'' $\TFNP$ subclasses defined by Papadimitriou \cite{Pap94} (the third being $\PPP$).
The seminal work of Beame et. al. \cite{Beame95} has shown that the important $\TFNP$ subclass $\PPA$, embodying the \emph{Handshaking Lemma}, is not contained in $\PPP$ in the black-box setting.
The recent breakthrough of \cite{Goeoes2022} showed that the $\TFNP$ subclass $\PLS$, embodying problems with (not necessarily efficient) local search heuristics, is \emph{also} not contained in $\PPP$ in the black-box setting.
Using our techniques, we improve these results, showing that neither of these classes are even contained in $\PAP$ --- the highest level of the Pecking Order:

\begin{restatable}{theorem}{plspap}
    \label{thm:plspap}
    $\PLS^{dt} \not\subseteq \PAP^{dt}$ and $\PPA^{dt}\not \subseteq \PAP^{dt}$.
\end{restatable}

To prove these results we crucially employ our proof complexity techniques, combined with a generalization of the technical notion of \emph{gluability} \cite{Goeoes2022}, used in the previous black-box separations between $\PPP$ and $\PLS$.
We refer to \cref{sec:other-classes} for technical details.

\subsubsection{Quantum Complexity and the Pecking Order}

Finally, we make an observation relating the Pecking Order to a problem in quantum complexity.
It has remained an open problem to get a separation between $\BQP$ and $\BPP$ with respect to a \emph{random oracle}, explicitly posed by Aaronson and Ambainis \cite{Aaronson2011}. They even identified a technical barrier to proving such a separation, called the Aaronson-Ambainis conjecture, which has evaded attacks by experts in Boolean function analysis for almost a decade. 

Recently, a breakthrough result was shown in this area by Yamakawa-Zhandry \cite{Yamakawa2022} who obtained such a separation for \emph{search} problems, evading the technical barrier that existed for decision problems. They also noted that their problem can be modified to show a separation between quantum and randomized query complexity in \TFNP. 
We prove the totality of their problem using the averaging principle, and thus place it in \PAP. This gives the first inclusion of the Yamakawa-Zhandry's Problem in a natural subclass of \TFNP.

\begin{restatable}{theorem}{yzpap}
    \label{lemma:yz}
    Yamakawa-Zhandry's Problem is contained in \PAP.
\end{restatable}

Furthermore, the proof of \cref{lemma:yz} suggests that the Yamakawa-Zhandry's problem necessarily corresponds to finding a $\poly(n)$-collision. Given our structure theorem of the Pecking Order (\cref{thm:general_lb}), one may suspect that Yamakawa-Zhandry's problem is not contained in any lower level of the Pecking Order, e.g., \PPP.

Finally, we remark that \PAP is a rather loose upper bound for the Yamakawa-Zhandry's problem. The \emph{structure} of the Yamakawa-Zhandry's problem, imposed by a specific choice of \emph{error-correcting code} (ECC) in its definition, is lost during our reduction. The structure of the ECC is essential for its quantum speed-up, considering the fact that even \PWPP cannot be solved efficiently by the quantum algorithm in the black-box setting~\cite{AS04}.

\subsection{Our Techniques}
\label{sec:our-techniques}

Our main separation results all follow from our black-box separations between $\ourPigeon{(t+1)}{M}{N}$ and $\ourPigeon{t}{M'}{N'}$, which in turn are proved by developing lower bound tools in \emph{propositional proof complexity}.
The connection between logic and \TFNP has long been acknowledged \cite{Buss94,Maciel00,Morioka01,Thapen02,Buresh-OppenheimM04,Jer09,KLNT11,BussJohnson12,Skelley11,BeckBuss14,Jer16,Thapen22}, and the first paper to prove oracle separations for \TFNP \cite{Beame95} also invoked a Nullstellensatz lower bound for their result. 
Recently, \emph{equivalences} have been established between complexity measures of certain proof systems and the complexity of black-box reductions to corresponding subclasses of \TFNP. 
The first example of this was G\"o\"os, Kamath, Robere and Sokolov \cite{GPRS19}, followed by G\"o\"os et al.~\cite{Goeoes2022} who proved equivalences for the prior natural-studied subclasses of $\TFNP$ and used these characterizations to provide the final remaining black-box separations between these subclasses.
This was followed by Buss, Fleming and Impagliazzo \cite{Buss23} who showed that this equivalence holds for every \TFNP problem --- in a certain sense --- although the corresponding proof system they construct is somewhat artificial, and hard to analyze directly.

Our techniques follow in this vein, as our lower-bound tools are directly inspired from propositional proof complexity.
However, a major place where we depart is that there is currently no known natural proof system characterizing \emph{any} of the classes in the Pecking Order, including $\PPP$! 
This means that we do not have any lower-bound tools from proof complexity to borrow directly, and must instead develop new ones.

The tools we develop to prove our lower bounds are generalizations of \emph{pseudoexpectation operators}, which have been instrumental in proving lower bounds for both the Sherali-Adams and the Sum-of-Squares proof system \cite{Fleming19}. 
Roughly speaking, a pseudoexpectation operator is an operator on polynomials that is indistinguishable from a true expectation operator, provided that we are only allowed to examine \emph{bounded moments} of the distribution. In particular, it is possible to define pseudoexpectation operators that can \emph{fool} a bounded adversary into thinking that there is a distribution over inputs to a total search problem $R$ that do not witness any solutions (which is, of course, absurd). These operators are instrumental for understanding both Sherali-Adams and the Sum-of-Squares proof system, and have broad connections to the theory of approximation algorithms (see \cite{Fleming19} and the references therein for further details). Thus, our paper provides a new link between pseudoexpectation operators and the complexity of the \emph{pigeonhole principle}.

Concurrent work of Fleming, Grosser, Pitassi, and Robere \cite{FlemingGPR23} introduced a new type of pseudoexpectation operator called a \emph{collision-free pseudoexpectation} that is tailored for proving lower bounds against black-box $\PPP$.
Collision-freedom is a technical condition that is difficult to summarize without introducing a significant amount of notation (cf.~\cref{sec:lb:collisions}).
For now, we will say that the original notion of collision-free pseudoexpectation operators only applied to reductions to $\ourPigeon{2}{N+1}{N}$ — that is, only for collisions of two pigeons, in the tight compression regime from $N+1$ pigeons to $N$ holes.
In our work, we introduce a broad generalization of collision-freedom that works for both an \emph{arbitrary} number of collisions \emph{and} an arbitrary compression rate.
All of our main lower bound results follow designing generalized collision-free pseudoexpectations, combined with reductions placing various other problems inside the Pecking Order. 
To be specific, for separations \emph{between} classes inside the Pecking Order (cf.~\cref{thm:general_lb}), we design collision-free pseudoexpectations directly and prove that they satisfy the required properties.
For our separations between the classes $\PLS$ and $\PPA$, we provide a more general approach.
Instead, we consider a general property (\emph{gluability}, introduced by \cite{Goeoes2022}), and show that any pseudoexpectation for a problem satisfying this property is \emph{automatically} collision-free for very strong parameters.
This is a strong technical improvement (and conceptually different) application of gluability than had been applied by the earlier work of \cite{Goeoes2022}.
We refer to \cref{sec:lb:collisions} and \cref{sec:other-classes} for further details on our pseudoexpectation technique.

\subsection{Related Works}\label{sec:concurrent}
In this section we discuss the relationship of our work with concurrent work. We also discuss the how our black-box separations fit in with related work in cryptography.

\paragraph{Relationship with the work of Fleming, Grosser, Pitassi, and Robere \cite{FlemingGPR23}.} In \cref{sec:lb:collisions} we introduce a new type of pseudoexpectation operator called a \emph{$(d, t, \varepsilon)$-collision-free pseudoexpectation} (cf.~\cref{def:collision-free}), and show that the construction of such operators implies lower bounds for $t$-$\PPP^{dt}$. This definition is a generalization of the notion of collision-free pseudoexpectation operators introduced in the concurrent work of Fleming, Grosser, Pitassi, and Robere. (To be precise, the original notion of a collision-free pseudoexpectation corresponds to the parameter setting $(d, 2, 1)$ in the above definition). The lower bounds presented in this paper are orthogonal to the results of \cite{FlemingGPR23}. The main result in the concurrent work is to give a black-box separation between $\PPP^{dt}$ and its Turing closure $\FP^{\PPP^{dt}}$, a result which is incomparable to the results proved in our work.

\paragraph{Relationship with the work of Li \cite{Li23tfap}.}

This work precedes the work of Li~\cite{Li23tfap}, which generalizes the proof of \cref{thm:lb:compression} into a framework that separates two types of \TFNP problems. In particular, \cite{Li23tfap} defines a new (semantic) subclass of \TFNP called \TFAP to captures \TFNP problems with \emph{abundant} solutions. 
\cite{Li23tfap} also introduces a notion called \emph{semi-gluablility}. The main result in \cite{Li23tfap} is that there is no black-box reduction from any ``semi-gluable'' problems to any \TFAP problems, and such separation could be extended to randomized reductions in most cases. 

In this paper, we present a direct proof of \cref{thm:lb:compression} in \cref{sec:lb:compression}, which is simpler than the proof of the main statement in Li~\cite{Li23tfap}, while sharing several key ideas with Li~\cite{Li23tfap}. 
Related work in bounded arithmetic has been done by M{\"{u}}ller~\cite{Muller21}.

\paragraph{Relationship with the work of Bennett, Ghentiyala, and Stephens{-}Davidowitz~\cite{BGS24merrier}.} The concurrent work \cite{BGS24merrier} defines the classes in the Pecking Order using different notation, and prove results including black-box separations. We view our black-box separations (\cref{thm:general_lb}) as conceptually stronger, since the statement does not depend on the compression rate of the two problems. In contrast,
the black-box separation from Bennett et al. builds upon the work of Rothblum \& Vasudevan \cite{RN22} and their technique only works with certain compression rates. For example, their technique cannot be used to infer lower bounds for $\ramsey$. They also study the relationship between the Pecking Order and certain coding and lattice problems, which has no overlap with our work.

\paragraph{The cryptographic picture.}

The search problems corresponding to the weak pigeonhole principles naturally show up in cryptography: given a hash function which compresses the input, how hard is it to find a collision? This is known as collision resistance hash function (CRH), which essentially corresponds to the average-case hardness of class \PWPP . Recently there has been work investigating a generalization of this to \emph{multi-collision resistance hash function (MCRH)} \cite{Joux04Multicollision, Bit18, KMY18, RN22}, which naturally corresponds to $n$-\PWPP defined in the Pecking Order. We refer readers to \cite{BGS24merrier} for a more comprehensive list of references.

It has remained open to show that there is no \emph{fully black-box construction}\footnote{Roughly speaking, a \emph{fully black-box construction} of primitive $A$ using primitive $B$ in the cryptography setting is a uniform randomized black-box Turing reduction from the computational problem of $B$ to the computational problem of $A$ in average-case.} of CRH using MCRH.\footnote{To the best of our knowledge, this is still an open problem. We were notified by the authors of \cite{KMY18} that there was an error in their proof that claimed to give a fully black-box separation between MCRH and CRH.}. Our separation results in the Pecking Order (\cref{thm:general_lb}) make the first step towards showing a \emph{fully black-box separation} between MCRH and CRH.

Our randomized separation of $\PPP$ from $n$-$\PWPP$ (\cref{thm:lb:compression}) may also be of interest from a cryptography perspective.
Berman et al.~\cite{BDRV18} and Komargodski et al.~\cite{KMY18} showed that there is no fully black-box construction of MCRH (corresponds to $t$-\PWPP)
using one-way permutation (corresponds to \PPP). Technically, their separation result is not comparable to ours.

Both \cite{BDRV18} and \cite{KMY18} use an indirect approach to rule out the fully black-box construction. In particular, they present a fully black-box construction of \emph{constant-round} statistically hiding commitment schemes using MCRH, and combine with a lemma from Haitner et al.~\cite{HHRS15} showing that there is no fully black-box construction of constant-round statistically hiding commitment schemes from one-way permutation. In contrast, our proof is a direct one, which may offer more insights into the distinction between the one-way permutation and the MCRH.

\subsection{Open problems}


In our opinion, a fine-grained study of extremal combinatorics problems with respect to the two generalizations of \PPP (\PLC and Pecking Order) merits further research. The connections to cryptography, and quantum complexity also raise several intriguing questions.
Here we list some open problems.


\begin{enumerate}
    \item Can we strengthen our black-box separations to further rule out Turing reductions? Meta-mathematically, this could be interpreted as the (standard) pigeonhole principle is not strong enough for proving the Ramsey theorem.\footnote{For the interested reader, we can state this formally in the language of bounded arithmetic as follows: $\ramsey$ is not uniformly Turing-reducible to $\pigeon$ in the black-box model if and only if Ramsey's Theorem is logically independent from the Pigeonhole Principle over the bounded arithmetic theory $\forall\mathsf{S}_2(\mathsf{PV}(\alpha))$ \cite{Muller21}.}
    \item Can we show a black-box separations between \PAP  and \PLC? This would separate the power of the ``iterated pigeonhole'' argument captured by $\PLC$ from the generalized pigeonhole principle.
    \item Does \ramsey lie in \PAP? Recall that with a slight loss in parameter \biramsey lies in \PAP (\cref{lem:biramsey-pap}).
    \item The definition of \PWPP is robust to different compression rate~\cite{Jer16}, e.g., both $\ourPigeon{2}{2N}{N}$ and $\ourPigeon{2}{{N^2}}{N}$ are \PWPP-complete. However, this is not known for  $t$-\PWPP for any $t > 2$~\cite{Joux04Multicollision}. Partial progress has been made in \cite{Bit18, BGS24merrier}.
    Either prove that these classes collapse under different compression rates or prove black-box separations.
    \item Is $\UPLC = n\textrm{-}\PWPP$? Given \cref{lem:pwpp-uplc,lem:uplc-pap}, this would be true if $n$-\PWPP is robust for a certain range of compression rates.
    \item Can we extend our separation results in Pecking Order to a fully black-box separation between the multi-collision-resistant hash function (MCRH) and the standard collision-resistant hash function (CRH)?
    \item Is there a natural \TFNP subclass --- simpler than the Yamakawa-Zhandry's problem --- in the Pecking Order which is contained in Total Function \BQP?
\end{enumerate}

\paragraph{Paper Organization}

In \cref{sec:prelim} we introduce necessary preliminaries for the paper, including definitions of the main $\TFNP$ search problems under consideration, the necessarily definitions for black-box $\TFNP$, and the proofs of two basic structural properties of the Pecking Order.
In \cref{sec:ramsey}, we characterize the complexity of $\ramsey$ and $\biramsey$ by proving \cref{thm:ramsey-pigeon} and relating them to the Pecking Order.
Then, in \cref{sec:lb:collisions}, we introduce and develop our lower-bound technique of $(d,t,\varepsilon)$-collision-free pseudoexpectations, and use it to prove \cref{thm:general_lb}, the main structure theorem of the Pecking Order.
In \cref{sec:lb:compression}, we prove the randomized separations between $\PPP$ and any problem in Pecking Order that has non-trivial compression rate.
Finally, \cref{sec:other-classes} proves our other inclusions in and separations from the Pecking Order, including proving that $\PLS^{dt}, \PPA^{dt} \not \subseteq \PAP^{dt}$, our results relating to $\PLC$ and $\UPLC$, and also our results about the Yamakawa-Zhandry's problem.

\section{Preliminaries}\label{sec:prelim}

Unless stated otherwise, $N = 2^n$. We also assume all the integer-valued functions defined in this paper are monotone and can be calculated efficiently.

\subsection{Ramsey Theory}\label{sec:prelim:ramsey}

Ramsey theory is most generally defined to be a branch of combinatorics that studies how large some structure must be such that a property holds, which is usually the presence of a substructure. This branch was started by the study of the \emph{Ramsey number}, which we define below.

\begin{definition}
    $R(s,t)$ is defined to be the smallest integer $n$ such that for any graph on $n$ vertices, there is an independent set of size $s$ or a clique of size $t$.
\end{definition}

\begin{theorem}[Ramsey's theorem \cite{Ramsey1929}]
    $R(s,t)$ is finite for every pair of integers $s,t$. 
\end{theorem}

Finding upper and lower bounds on $R(s,t)$ has been a big program in the combinatorics community. But recently there has been interest in building \emph{explicit} graphs with no clique or independent set of size $K$ (which witness lower bounds for $R(K,K)$). These are motivated by connections to pseudorandomness. Constructing a pseudorandom object called a two-source \emph{disperser} is essentially equivalent to constructing explicit bipartite Ramsey graphs, and getting better parameters has seen a long line of work \cite{ChorGoldreich88,Cohen16,Chattopadhyay19,Li2023}. This program makes progress on answering an old question of Erd\H{o}s \cite{Erds1947SomeRO}: can we construct explicit $O(\log n)$-Ramsey graphs on $n$ vertices? We formally define this object below.

\begin{definition}[$K$-Ramsey]
    A graph on $N$ vertices is said to be $K$-Ramsey if it does not contain a clique or independent set of size $K$.
\end{definition}

\begin{definition}[Bipartite $K$-Ramsey]
    A bipartite graph on $(N,N)$ vertices is said to be $K$-BiRamsey if it does not contain a biclique or independent set of size $(K,K)$.
\end{definition}

We recall a recent result of Li \cite{Li2023}.

\begin{theorem}[\cite{Li2023}]
\label{thm:explicit-ramsey}
    There exists a constant $c > 1$ such that for every integer $N$ there exists a (strongly) explicit construction of a bipartite $K$-Ramsey graph on $N$ vertices with $K = \log^c N$.
\end{theorem}

We now define \TFNP problems corresponding to these extremal combinatorics results.

\begin{definition}[$K$-$\ramsey$]
    For this problem, the input specifies a graph on $N = 2^n$ vertices; the parameter $K(n)$ is a function satisfying $R(K,K) \leq N$.
    \begin{description}
        \item[Input] A pair $(n,C)$, where $C: \B^n \times \B^n \mapsto \B$ is a circuit of $\poly(n)$ size.
        \item[Solutions] We have two kinds of solutions. One is a certificate that $C$ is not a valid {\bf encoding} of a simple undirected graph: $u$ for which $C(u,u)=1$ (self-loops) or $u,v$ for which $C(u,v) \neq C(v,u)$ (directed edge). Else, we want to find $K$ indices $V = \{v_1, v_2 \ldots v_{K}\}$ which form a {\bf clique} or {\bf independent set}.
    \end{description}
\end{definition}

Note that we must choose $K$ such that $R(K,K) \leq N$ in order for the problem to be total. An analogous condition holds for the bipartite analogue defined below below. We use $R_b(s,t)$ to denote the bipartite analogue of the Ramsey number.

\begin{definition}[$K$-$\biramsey$] For this problem, the input specifies a bipartite graph with $N$ vertices on each side, where $N = 2^n$; the parameter $K(n)$ is a function satisfying $R_b(K,K) \leq N$.

\begin{description}
    \item[Input] A pair $(n,C)$, where $C: \B^n \times \B^n \mapsto \B$ is a circuit of $\poly(n)$ size.
    \item[Solutions] $(K,K)$ indices $U = \{u_1, u_2 \ldots u_{K}\}$ and $V = \{v_1, v_2 \ldots v_{K}\}$ such that $(U, V)$ forms either a {\bf biclique} or an {\bf independent set}.
\end{description}
\end{definition}

We drop the parameter $K$ when we mean $K=n/2$. This is the existence result (up to logarithmic factors) guaranteed by the original theorem by Ramsey~\cite{Ramsey1929}. Note that recently the first constant factor improvement in this existential result was obtained by Campos et al.~\cite{campos2023exponential}, which corresponds to an exponential improvement in the Ramsey number. We also omit the size parameter $N$ when it is clear in the context.

\subsection{Decision Tree \TFNP}\label{sec:prelim:tfnp_dt}

\begin{definition}
 A \emph{query total search problem} is a sequence of relations $\textsc{R} = \set{R_n \subseteq \B^{n} \times O_n}$, where $O_n$ are finite sets, such that for all $x \in \B^n$ there is an $o \in O_n$ such that $(x, o) \in R_n$.
 A total search problem $\textsc{R}$ is in $\TFNP^{dt}$ if for each $o \in O_n$ there is a decision tree $T_o$ with depth $\poly(\log n)$ such that for every $x \in \B^n$, $T_o(x) = 1$ iff $(x, o) \in \textsc{R}$.
\end{definition}

In general, we use $DT(\textsc{R})$ to denote the query complexity of a search problem $\textsc{R}$.
To simplify the presentation and the relationship between the black-box model and the white-box model, we adhere to several conventions:
\begin{itemize}
    \item For problems with a non-binary input alphabet, we simulate it with the usual binary encoding. For instance, in the pigeonhole principle, a mapping of a pigeon to a hole (i.e.~a pointer in the range $[n]$) can be simulated by a $\log n$ bit binary encoding. All problems discussed in this paper have $\poly(n)$ alphabet size, and can therefore be simulated with $O(\log n)$ overheads.
    \item The problem $R_n$ is permitted to have input bits on the order of $\poly(n)$.
    \item We sometimes abuse the notation by calling an individual relation $R_n$ a search problem, rather than a whole sequence $\mathrm{R} = (R_n)$.
    \item The superscript “dt” is omitted when referring to $\TFNP^{dt}$ problems if the context makes it clear.
\end{itemize}

For instance, when $\ourPigeon{t}{M}{N}$ problem is defined in the black-box model, the polynomial-size circuit $\pmap$ encoding the mapping of pigeons to holes is replaced by an oracle that queries the pigeons and maps them to holes.
In the black-box model we also have an appropriate version of reducibility between search problems, where the reductions are computed by low-depth decision trees.
We introduce this next.

\begin{definition}[Decision Tree Reduction]
    A \emph{decision tree reduction} from relation $R \subseteq \{0, 1\}^n \times O$ to $Q \subseteq \{0, 1\}^m \times O'$ is a set of depth-$d$ decision trees $f_i: \{0, 1\}^n \rightarrow \{0,1\}$ for each $i \in [m]$ and $g_o: \{0,1\}^n \rightarrow O$ for each $o \in O'$ such that for any $x \in \{0,1\}^n$,
    \[(x, g_o(x)) \in R  \Leftarrow (f(x), o) \in Q,\]
    where $f(x) \in \{0,1\}^m$ has $f_i(x)$ as the $i$-th bit.
    The \emph{depth} of the reduction is $d$, and the \emph{size} of the reduction is $\log m$.
    The \emph{complexity} of the reduction is $d + \log m$, and we write $Q^{dt}(R)$ to denote the minimum complexity of a decision-tree reduction from $R$ to $Q$, or $\infty$ if one does not exist.

    We extend these notations to sequences in the natural way.
   	If $R$ is a single search problem and~$\textsc{Q} = (Q_m)$ is a sequence of search problems, then we denote by $\textsc{Q}^{dt}(R)$ the minimum of $Q^{dt}_m(R)$ over all $m$.
	If $\textsc{R} = (R_n)$ is also a sequence, then we denote by $\textsc{Q}^{dt}(\textsc{R})$ the function $n\mapsto \textsc{Q}^{dt}(R_n)$. 
    A $\TFNP^{dt}$ problem $\textsc{R}$ can be \emph{black-box reduced} to $\TFNP^{dt}$ problem $\textsc{Q}$, written $\textsc{R} \leq_m \textsc{Q}$, if there is a $\poly(\log(n))$-complexity decision-tree reduction from $\textsc{R}$ to $\textsc{Q}$.
\end{definition}

In general, if a syntactic $\TFNP$ subclass $\cA$ is defined by polynomial-time black-box reductions to a complete problem $\pA$, we will use $\cA^{dt}$ to denote the class of query total search problems that can be efficiently black-box reduced to $\pA$.
The next theorem motivates the decision-tree setting: constructing separations in the decision tree setting implies the existence of generic oracles separating the standard complexity classes.

\begin{theorem}[\cite{Beame95}, Informal]\label{thm:beame}
    For two syntactical \TFNP subclasses $\mathsf{A, B}$, $\mathsf{A}^{dt} \nsubseteq \mathsf{B}^{dt}$ \emph{implies} the existence of a (generic) oracle $O$ separating $\mathsf{A}$ from $\mathsf{B}$, i.e., there is no black-box reduction from $\mathsf{A}$ to $\mathsf{B}$.
\end{theorem}

We will also deal with \emph{randomized} reductions in \TFNP. We define these formally below.

\begin{definition}[Randomized Decision Tree Reduction]\label{def:rand_reduct}
    A \emph{randomized decision tree reduction} from relation $R \subseteq \{0, 1\}^n \times O$ to $Q \subseteq \{0, 1\}^m \times O'$ is a distribution $D$ of depth-$d$ decision tree reductions from $R$ to $Q$, such that for any $x \in \{0,1\}^n$, 
    \[\Pr_{((f_i), (g_o)) \sim D } [(x, g_o(x)) \in R  \Leftarrow (f(x), o) \in Q] \geq \frac{1}{2}.\]
\end{definition}

\subsection{Basic Structural Properties of the Pecking Order}\label{sec:prelim:PO}

\Xcomment{
    \begin{lemma}\label{lem:hierarchy}
        $t$-\PPP $\subseteq$ $(t+1)$-\PPP. More generally, $a\textrm{-}\PPP \subseteq b\textrm{-}\PPP$ if $a(n) \leq b(n)$ for all $n$.
    \end{lemma}
    
    \begin{proof}
        Given a $\tightPigeon{t}$ instance $(n, h)$ ($(t-1)N+1$ pigeons and $N$ holes), we reduce it to a $\tightPigeon{(t+1)}$ instance $(n, h')$ by adding $N$ dummy pigeons. The $i$-th dummy pigeon is mapped to the $i$-th hole for any $i \in [N]$; the mapping for other pigeons is left unchanged. It is easy to verify that any $(t+1)$-collision in $h'$ is a $t$-collision of $h$. 
    
        The general case follows from the same trick of adding dummy pigeons.
    \end{proof}
}




By a simple padding argument, $t$-\PPP forms a hierarchy when $t$ increases.

\begin{lemma}\label{lem:hierarchy}
    $t$-\PPP $\subseteq$ $(t+1)$-\PPP. More generally, $a\textrm{-}\PPP \subseteq b\textrm{-}\PPP$ if $a(n) \leq b(n)$ for all $n$.
\end{lemma}

\begin{proof}
    Given a $\tightPigeon{t}$ instance $(n, h)$ ($(t-1)N+1$ pigeons and $N$ holes), we reduce it to a $\tightPigeon{(t+1)}$ instance $(n, h')$ by adding $N$ dummy pigeons. The $i$-th dummy pigeon is mapped to the $i$-th hole for any $i \in [N]$; the mapping for other pigeons are left unchanged. It is easy to verify that any $(t+1)$-collision in $h'$ is a $t$-collision of $h$. 
    
    The general case follows from the same trick of adding dummy pigeons.
\end{proof}

The same proof can be generalized to show that $a$-\PWPP $\subseteq$ $b$-\PWPP when $a(n) \leq b(n)$. 




\PolynomialClose*

\PolynomialRobustness*

\begin{proof}
    Let $p(n)$  be a polynomial such that $b(n) \leq a(p(n))$ for all $n$. Let $c(n) \coloneqq a(p(n))$. \cref{lem:hierarchy} implies that $b$-\PPP $\subseteq$ $c$-\PPP. By the symmetry of the statement, it suffices to show that $c$-\PPP $\subseteq$ $a$-\PPP.

    Take a $c$-\pigeon instance $x \coloneqq (n,\pmap)$. For technical convenience, we use an alternate encoding where $M = N \cdot (c(n)-1)$, and the goal is to find either $(c(n)-1)$ pigeons mapped to $1$ or a $c(n)$-collision. It is easy to show the equivalence between this alternate version and the original version.
    
    We now construct an $\tightPigeon{a}$ instance $x' \coloneqq (h', n')$, where $n' = p(n)$.  
    Let $N' \coloneqq 2^{n'}, M' \coloneqq a(n') \cdot N'$ be the number of holes and pigeons that are supposed to be in $x'$. $x'$ is simply chosen as $N'/N$ disjoint copies of $x$ in parallel.
    It is easy to verify that the number of holes in $x'$ is $N \cdot (N'/N) = N'$, and the number of pigeons in $x'$ is \[ M \cdot (N'/N) = (c(n)-1) \cdot N' = (a(n')-1) \cdot N' = M'.\]

    Formally, we index each hole in $x'$ by a pair $(i,l) \in [N] \times [N' / N]$; each pigeon in $x'$ is also indexed by a pair $(j,l) \in [M] \times [N' / N]$.
    Now, $\pmap' : [M] \times [N' / N]  \mapsto [N] \times [N' / N]$  is defined as \[
        \pmap'(j, l) := (\pmap(j), l), \forall (j,l) \in [M] \times [N' / N].
    \]

    Note that any collision in $x'$ must come from a single copy of $x$, because different copies of $x$ are disjoint. Also given that $a(n') = c(n)$, any solution of $x'$ immediately corresponds to a solution of $x$. This concludes the correctness of our reduction.

\end{proof}

One can easily generalize the proof of \cref{thm:polynomial-robustness} to show that $a\textrm{-}\PWPP = b\textrm{-}\PWPP$ when $a(n), b(n)$ are polynomially close. Note that all polynomial functions are polynomially close to each other. Therefore, $t$-\PPP is equivalent to $\PAP$ for any choice of $t(n) = \Theta(\poly(n))$.

\section{Ramsey and the Pecking Order}
\label{sec:ramsey}

In this section, we show that problems in the Pecking Order can be reduced to $\ramsey$ and $\biramsey$ (\cref{lem:pigeon2ramsey}).
As a consequence of the generality of parameters in the structural theorem of the Pecking Order (\cref{thm:general_lb}), we immediately get that $\ramsey^{dt}$ is not in  $\PPP^{dt}$ (\cref{thm:ramseyinstantiated}). In the reverse direction, we show that $\biramsey$ with a slightly weaker parameter is included in $\PAP$ (\cref{lem:biramsey-pap}); combined with \cref{thm:ramseyinstantiated}, we get an almost tight characterization of $\biramsey$ using the Pecking Order.

We start with reducing problems in the Pecking Order to \ramsey, which directly generalize a reduction that appeared in \cite{Komargodski2017}.
For ease of discussion, we only state the following result in terms of \ramsey, yet the same argument also holds true for \biramsey.

\begin{lemma}[Generalization of \protect{\cite[Theorem~3.1]{Komargodski2017}}]
    \label{lem:pigeon2ramsey}
    Suppose there exists a graph on $N$ vertices with no $K$ clique or independent set. Then $\ourPigeon{t}{M}{N}$ can be black-box reduced to $t(K-1)$-$\ramsey$ when $2t(K-1) \le \log M$.
\end{lemma}

\begin{proof}
    Given a graph $G_0(V_0, E_0)$ on $N$ vertices such that it does not contain a $K$ clique or independent set, we will build an instance $G$ given by $(\log M,E)$ of $t(K-1)$-$\ramsey$ from an instance $C$ of $\ourPigeon{t}{M}{N}$.

    Let $\pmap : [M] \mapsto [N]$ be an instance of $\ourPigeon{t}{M}{N}$. We will define a $t(K-1)$-$\ramsey$ instance $G$ using the \emph{graph hash product} from \cite{Komargodski2017}. Let $G = G_0 \otimes \pmap = (V, E)$ be a graph on $V = [M]$ such that \[
        (u, v) \in E \iff \pmap(u) = \pmap(v) \text{ or } (\pmap(u), \pmap(v)) \in E_0.
    \]
    
    Now we prove that the solution $S$ returned by $t(K-1)$-$\ramsey$ witnesses a $t$-collision in $\pmap$. Let $S'$ be set $\{\pmap(u) \mid u \in S\} \subseteq V(G_0)$. By the definition of $G$, $S'$ is a clique if $S$ is a clique, and $S'$ is an independent set if $S$ is an independent set. Therefore, $S'$ is either a clique or an independent set. 

    Since $G_0$ does not contain a $K$ clique or independent set, we have $|S'| < K$. Given $|S| = t(K-1)$, by an averaging argument $S$ must witness a $\frac{t(K-1)}{K-1} = t$ collision in $\pmap$.

\end{proof}

Using the probabilistic method, Erd\"{o}s \cite{Erds1947SomeRO} shows that there exists a graph on $N$ vertices with no clique or independent set of size $K = 2\log N$. This gives us the following theorem as a corollary:

\ramseyinstantiated*
We also note that by \cref{thm:explicit-ramsey}, we have that for some $c > 1$, we have an explicit efficient black-box reduction from ${t\text{-}\pigeon_N^M}$ to ${\ramsey}$ whenever $2t(\log^c N-1) \le \log M$. As a consequence of these reductions, we get two corollaries of \autoref{thm:general_lb} regarding the place of $K$-\ramsey in \TFNP. Below, we assume \cref{thm:general_lb}, which is proved in \cref{sec:lb:collisions}.

\ramseyppp*
\begin{proof}

\autoref{thm:ramseyinstantiated} shows that there exists some polynomial $p(n)$ such that $\ourPigeon{p(n)}{M}{N}$ reduces to $\ramsey_M$ for large enough $M$. However, by \autoref{thm:general_lb}, we know that there's no black-box reduction from $\ourPigeon{p(n)}{M}{N}$ to $\ourPigeon{t(n)}{M'}{N'}$ for any subpolynomial $t(n)$. This implies that \ramsey is not in $t(n)$-$\PPP^{dt}$ for any subpolynomial $t(n)$.
\end{proof}

Actually, by being a little more careful with the parameters, we can prove a stronger version of the theorem following the same argument, which shows that even $\log^c(N)\text{-}\ramsey^{dt}$ 
 is not in $\SAP^{dt}$.

        \Xcomment{
        \begin{theorem}
            For every constant $0 < c < 1$, $\log^c(N)\text{-}\ramsey^{dt} \notin \SAP^{dt}$.
        \end{theorem}
        
        \begin{proof}
        \autoref{thm:ramseyinstantiated} actually shows that $\ourPigeon{n}{M}{N}$ reduces to $n(2n-1)\text{-}\ramsey$ with an instance size $\log M = n^{3/c}$. In this parameter setting, $n(2n-1)\text{-}\ramsey$ could be further reduced to $\log^c(M)\text{-}\ramsey$. \sid{I find this hard to parse.}
        
        Therefore, we can finish the proof by invoking \autoref{thm:general_lb} which shows there is no efficient black-box reduction from $\ourPigeon{n}{M}{N}$ to $\ourPigeon{t(n)}{M'}{N'}$ for any subpolynomial $t(n)$.
        \end{proof}
        
        We remark again that the theorems presented above also work for \biramsey.}

Finally, we show that with a slight loss in the parameter, \biramsey fits into the Pecking Order. 

\BiRamseyinPAP*
\begin{proof}
    Given a bipartite graph $G = ([N], [N], E)$, we say $E(x, y) = 1$ if $(x, y) \in E$ and $E(x, y) = 0$ otherwise. We construct an $n\text{-}\pigeon^{N}_{N/n}$ instance using function $h : [N] \mapsto \B^{n - \log n}$ defined by \[
        h(y) := (E(x, y))_{x \in [n - \log n]}.
    \]
    Our goal is to prove that we can efficiently find a clique or independent set of size $(n - \log n)/{2}$ from an $n$-collision in $h$. Let $y_1, \dots, y_n \in [N]$ be an $n$-collision in $h$. Then by the definition of $h$, we have for each $x \in [n - \log n]$, $E(x, y_1) = E(x, y_2) = \dots = E(x, y_n)$. 
    
    At least half of the $x$'s in $[n - \log n]$ will have the same value for $E(x, y_1)$, and we let these indices be $x_1 < \ldots < x_{(n - \log n)/{2}}$. This gives us that for all $i, j \in [(n - \log n) / 2]$, $E(x_i, y_j) = E(x_1, y_j) = E(x_1, y_1)$. Therefore, $(\set{x_1, \dots, x_{(n - \log n)/ 2}}, \{y_1, \dots, y_{(n - \log n)/ 2}\})$ is either a biclique or an independent set.

\end{proof}

\paragraph{Query complexity of \ramsey.} Besides the relative complexity of \ramsey, one might also ask questions about its query complexity in various models (deterministic, randomized, quantum). To find a clique or independent set of size $k$, it certainly suffices to query all the vertices of an arbitrary subgraph of size $R(k,k)$. This gives us an upper bound on the deterministic query complexity of $\ramsey$ 
of $\binom{R(n/2,n/2)}{2}$. Plugging in the best known upper bound on the diagonal Ramsey number \cite{campos2023exponential}, we get an upper bound on $N^{2-\epsilon}$ for a small constant $\epsilon$. On the lower bound front, we can infer that the quantum query complexity of $\ramsey$ is at least $N^{1 - o(1)}$ due to the reduction in \cref{lem:pigeon2ramsey} combined with the quantum query lower bound by Liu and Zhandry \cite{Liu19}. Since the best lower bound we have on $R(t,t)$ is $2^{\widetilde{O}(t/2)}$ \cite{Erds1947SomeRO}, improving even the deterministic lower bound significantly beyond $\Omega(N)$ would give a \emph{new combinatorics result}! The best deterministic lower bound known is due to Conlon et al.~\cite{Conlon19}. We note this as an exciting approach to getting better lower bounds for the Ramsey numbers. 

\section{Structure of the Pecking Order}\label{sec:lb:collisions}



In this section we prove \cref{cor:pih-sep}, our separation of the Pigeon Hierarchy, restated here for convenience.

\pihsep*

As we discussed in the introduction, our lower bounds are proved using (generalizations of) tools from propositional proof complexity, and in particular the theory of \emph{pseudoexpectation operators}.
Our main theorem is proved by generalizing the notion of \emph{collision-free pseudoexpectation operators}, designed by Fleming, Grosser, Pitassi, and Robere to give a black-box separation between $\PPP^{dt}$ from its Turing-closure, to the entire Pecking Order \cite{FlemingGPR23}.

\paragraph{Pseudoexpectation Operators.}

First we introduce the notion of a \emph{pseudoexpectation operator}, and for this we need to recall some basic results about multilinear polynomials.
Let $x_1, x_2, \ldots, x_n$ be a family of $\set{0,1}$-valued variables. 
We consider real-coefficient polynomials $p \in \RR[x_1, \dots, x_n]$ over these variables.
All polynomials that we consider are \emph{multilinear}, meaning the individual degree of any variable is at most $1$.
The algebra of multilinear polynomials is described by the quotient ring $\RR[x_1, \dots, x_n]/\langle x_i^2 - x_i\rangle_{i=1}^n$.
Formally, the addition of two multilinear polynomials is still a multilinear polynomial, but, if we multiply two multilinear polynomials then, after multiplication, we drop the exponents of all variables to $1$.
For example, $(x+y)(x+y) = x + 2xy + y$, as multilinear polynomials.

For any $S, T \subseteq [n]$ with $S \cap T = \emptyset$, define the polynomial \[ C_{S, T} := \prod_{i \in S} x_i \prod_{j \in T} (1-x_j). \] 
Note that for $\set{0,1}$-assignments the polynomial $C_{S, T}$ encodes the truth value of the conjunction $\bigwedge_{i \in S} x_i \land \bigwedge_{j \in T} \overline x_j$, and thus we will refer to $C_{S,T}$ as a ``conjunction'' in an abuse of notation.
If $R_n \subseteq \B^n \times O$ is a query total search problem, and $o \in O$, then a conjunction $C$ \emph{witnesses} the solution $o$ if $C(x) = 1 \Rightarrow (x,o) \in R_n$ for every $x \in \B^n$.
Similarly, we say a conjunction $C$ \emph{witnesses} $R_n$ if it witnesses some solution to $R_n$.
Each conjunction $C = C_{S,T}$ is naturally associated with a partial restriction $\rho(C) \in \set{0,1,*}$ defined by setting 
\[ \rho(C)_i = \begin{cases}
    1 & i \in S, \\
    0 & i \in T, \\
    * & \text{otherwise}.
\end{cases}\]

\begin{definition}
    \label{def:pseudoexpectation}
    Let $n \geq d$ be positive integers.
    Let $\mathcal{P}_{n,d}$ be the collection of all degree $\leq d$ multilinear polynomials over variables $x_1, \dots, x_n$.
    An operator $\psE: \mathcal{P}_{n,d} \rightarrow \RR$ is a \emph{degree-$d$ pseudoexpectation operator} if it satisfies the following three properties:
    \begin{itemize}
        \item {\bf Linearity.} $\psE$ is linear.
        \item {\bf Normalized.} $\psE[1] = 1$.
        \item {\bf Nonnegativity.} $\psE[C] \geq 0$ for all degree $\leq d$ conjunctions $C$.
    \end{itemize}
    Furthermore, if $R \subseteq \B^n \times O$ is a query total search problem, then $\psE$ is \emph{$R$-Nonwitnessing} if it additionally satisfies the following property:
    \begin{itemize}
        \item {\bf $R$-Nonwitnessing.} $\psE[C] = 0$ for any conjunction $C$ witnessing $R$.
    \end{itemize}
    If $\mathcal{F}$ is any sequence of degree $\leq d$ multilinear polynomials, then we write $\psE[\mathcal{F}] := \sum_{p \in \mathcal{F}} \psE[p]$.
\end{definition}

Pseudoexpectation operators were originally introduced to prove lower bounds on the degree of Sherali-Adams refutations for unsatisfiable CNF formulas \cite{Fleming19}.
Often the easiest way to construct a pseudoexpectation is to construct an object called a \emph{pseudodistribution} instead.
We introduce pseudodistributions next:

\begin{definition}
    Let $x_1, \ldots, x_n$ be a set of $\B$-valued variables, and let $d \leq n$.
    A \emph{degree-$d$ pseudodistribution} over these variables is a family of probability distributions \[ \mathcal{D} = \set{\mathcal{D}_S : S \subseteq [n], |S| \leq d}, \] such that the following properties hold:
    \begin{itemize}
        \item For each set $S \subseteq [n]$, $|S| \leq d$, $\mathcal{D}_S$ is supported on $\set{0,1}^S$, interpreted as boolean assignments to variables in $\set{x_i : i \in S}$.
        \item For each $S, T \subseteq [n]$, $|S|, |T| \leq d$, we have $\mathcal{D}_S^{S \cap T} = \mathcal{D}_T^{S \cap T} = \mathcal{D}_{S \cap T}$, where $\mathcal{D}_A^{B}$ for $B \subseteq A$ is the marginal distribution of variables in $B$ with respect to $\mathcal{D}_A$. 
    \end{itemize}
\end{definition}

The following standard lemma allows us to construct pseudoexpectations from pseudodistributions. In fact, the two objects are equivalent, but we will not need the converse construction in this paper.

\begin{lemma}[\cite{Fleming19}]
    Let $\mathcal{D}$ be a degree-$d$ pseudodistribution over variables $x_1, \dots, x_n$. The operator $\psE$ defined by \[ \psE\left[\prod_{i \in S} x_i\right] = \Pr_{y \sim \mathcal{D}_S}[\forall i \in S: y_i = 1]\] and extended to all multilinear polynomials by linearity, is a degree-$d$ pseudoexpectation.
\end{lemma}

The following pseudodistribution, and its accompanying pseudoexpectation, is the central pillar of all lower bounds in this paper.
This example is, in fact, one of the classical examples of a pseudodistribution \cite[Lemma 2]{Magen08}.

\begin{definition}[Matching Pseudodistribution]
    \label{def:matching-pseudodistribution}
    Consider $\ourPigeon{t}{M}{N}$ with $M \geq (t-1)N + 1$ pigeons and $N$ holes, and let $d \leq N/2$.
    The \emph{degree-$d$ matching pseudodistribution} for this instance is the following pseudodistribution.
    For any set of input variables $S$ in $\ourPigeon{t}{M}{N}$, let $p(S)$ denote the set of pigeons mentioned among variables in $S$.
    For each subset $S$ of variables with $|S| \leq d$, define the distribution $\mathcal{D}_S$ by sampling a uniformly random matching from the pigeons in $p(S)$ to $|p(S)|$ holes, and assigning the variables in $S$ according to this matching. Further define $\mathcal{D} = \set{\mathcal{D}_S : |S| \leq d}$ to be the collection of all such distributions.
\end{definition}

The above construction indeed defines a pseudodistribution, as shown in \cite{Magen08}.
To put it simply, if we sample a uniformly random matching from $t$ pigeons to $t$ holes and then marginalize one pigeon out, then the result is a uniformly random matching from $t-1$ pigeons to $t-1$ holes.

\paragraph{Lower Bounds for the Pecking Order.}
We are now ready to prove the main lower bound result of this paper. 
As we have mentioned above, degree-$d$ pseudoexpectation operators were originally introduced to prove lower bounds for Sherali-Adams refutations. 
In order to do this, we must construct high-degree pseudoexpectation operators which are additionally \emph{nonwitnessing}, meaning that $\psE[C] = 0$ whenever $C$ is a conjunction witnessing a solution to our search problem $R$.
A recent work \cite{Goeoes2022} has shown that a query total search problem $R$ is in $\PPADS^{dt}$ if and only if an unsatisfiable CNF formula $\cnf(R)$ expressing the totality of $R$ has low-degree unary Sherali-Adams refutations. 
This means that constructing a nonwitnessing pseudoexpectation for $R$ automatically implies that $R \not \in \PPADS^{dt}$.
However, $R$ can admit a nonwitnessing pseudoexpectation but still lie in $\PPP^{dt}$ or higher up in the Pecking Order --- for example, the matching pseudoexpectation is a nonwitnessing pseudoexpectation for $\tightPigeon{2}$, which is the complete problem for $\PPP$.

Therefore, to prove lower bounds for higher levels of the Pecking Order, we must strengthen the definition of a pseudoexpectation.
To introduce this strengthening, we need the following auxiliary definition.

\begin{definition}
    \label{def:witnessing-family}
    Let $R \subseteq \B^n \times O$ be a query total search problem, and let $\mathcal{F}$ be a family of degree-$d$ conjunctions over input variables of $R$.
    For $t \geq 2$, the family $\mathcal{F}$ is said to be \emph{$t$-witnessing for $R$} if for any subset $\mathcal{S} \subseteq \mathcal{F}$ with $|\mathcal{S}| = t$, either $\prod_{C \in \mathcal{S}} C \equiv 0$, or $DT(R \restriction \rho) \leq d$, where $\rho$ is the concatenation of $\rho(C)$ for all $C \in \mathcal{S}$, and $D$ is some universal constant.
    In other words, either every subset of $t$ conjunctions is inconsistent, or, there is a shallow decision tree solving the restricted problem $DT(R \restriction \rho)$, where $\rho$ is the union of partial assignments corresponding to the $t$ conjunctions. We say that $\mathcal{F}$ is \emph{$t$-witnessing} if the problem $R$ is clear from context.
\end{definition}

\begin{definition}
    \label{def:collision-free}
    Let $R \subseteq \B^n \times O$ be a total query search problem, let $d, t$ be positive integers, and let $\varepsilon > 0$ be a real parameter.
    Let $\psE$ be a degree $D \geq d$ pseudoexpectation operator.
    Then $\psE$ is \emph{$(d,t,\varepsilon)$-collision-free for $R$} if it satisfies the following property:
    \begin{itemize}
        \item {\bf $t$-Collision-Freedom.} $\displaystyle \psE[\mathcal{F}] \leq \varepsilon$, for every $t$-witnessing family $\mathcal{F}$ of degree $\leq d$ conjunctions.
    \end{itemize}
\end{definition}

The notion of a collision-free pseudoexpectation was introduced by Fleming, Grosser, Pitassi, and Robere \cite{FlemingGPR23} in the special case where $t = 2, \varepsilon = 1$, in order to separate $\PPP$ from its Turing closure in the black-box setting.
The above definition generalizes this notion to arbitrary size-$t$ collisions. 
As we will see, $t$-Collision Freedom is the additional property that is required of a pseudoexpectation in order to rule out membership of problems in $t$-$\PPP$.
The following theorem generalizes the same theorem for $M = N+1, t = 2, \varepsilon = 1$, proved by \cite{FlemingGPR23}.

\begin{theorem}
\label{thm:psExp}
    Let $R \subseteq \B^n \times O$ be a query total search problem.
    Let $M, N, t$ be positive integers with $M \geq (t-1)N + 1$, and let $0 \leq \varepsilon < M/N$ be any real parameter.
    If there is a $(d,t, \varepsilon)$-collision-free pseudoexpectation operator for $R$ then there is no depth-$d$ decision-tree reduction from $R$ to $\ourPigeon{t}{M}{N}$.
\end{theorem}
\begin{proof}
    For the sake of intuition, we first prove this for the case of $(d, 2, \varepsilon)$-pseudoexpectation operators, corresponding to the classical $\TFNP$ class $\PPP$, but the argument easily generalizes to arbitrary $t, N,$ and $M \geq (t-1)N + 1$.
    Let us assume by way of contradiction that there is a degree-$d$ decision-tree reduction from $R$ to $\ourPigeon{2}{N+1}{N}$ for some $N$, and let $\psE$ denote the $(d,2,\varepsilon)$-collision-free pseudoexpectation for $R$.
    Let $T_1, T_2, \ldots, T_{N+1}$ denote the depth-$d$ decision trees in the $\ourPigeon{2}{N+1}{N}$ instance produced by the reduction mapping the pigeons to holes.
    First, for any decision tree $T_i$ let $L(T_i)$ denote the leaves of $T_i$, and, for any leaf $\ell$ of $T_i$, let $C_\ell$ denote the conjunction obtained by multiplying the literals along the path to $\ell$.
    An easy induction on the depth of the tree shows that for every tree $T_i$, \[\sum_{\ell \in L(T_i)} \psE[C_\ell] = 1.\]

    Now for any hole $h \in [N]$, let $\mathcal{F}_h$ denote the set of all conjunctions $C_\ell$ that correspond to paths of any decision tree among $T_1, \dots, T_{N+1}$ such that the leaf $\ell$ of that path is labelled with $h$.
    Since this instance of $\ourPigeon{2}{N+1}{N}$ is obtained via a depth-$d$ reduction from $R$, for any pair of conjunctions $C, D \in \mathcal{F}_h$, either $C$ and $D$ are inconsistent, or, $DT(R \restriction \rho(CD)) \leq d$ since a collision of two pigeons implies that we can recover a solution of $R$ after at most $d$ more queries.
    It follows that the family $\mathcal{F}_h$ is $2$-witnessing, in the language of \cref{def:witnessing-family}.

    Since $\mathcal{F}_h$ is $2$-witnessing for each hole, it follows that $\psE[\mathcal{F}_h] \leq \varepsilon$ for every $h \in [N]$ since $\psE$ is $(d,2,\varepsilon)$-collision-free.
    Now, observe that \[\sum_{i=1}^{N+1} \sum_{\ell \in L(T_i)} \psE[C_\ell] = \sum_{h=1}^N \psE[\mathcal{F}_h],\] since each of the $N+1$ pigeons are mapped to exactly one hole under a total assignment to the variables.
    This means that \[ N+1 = \sum_{i=1}^{N+1} \sum_{\ell \in L(T_i)} \psE[C_\ell] = \sum_{h =1}^N \psE[\mathcal{F}_h] \leq \varepsilon N < N+1,\]
    which is a contradiction.

    To generalize this for arbitrary $t$, we instead consider reductions to $\ourPigeon{t}{M}{N}$ and substitute $(d,2,\varepsilon)$-witnessing with $(d,t,\varepsilon)$-witnessing in the above proof. Since the instance of $\ourPigeon{t}{M}{N}$ is obtained by depth-$d$ reduction from $R$, it follows now that for every set of $t$ distinct conjunctions $C_1, C_2, \ldots, C_t \in \mathcal{F}_h$, either $\prod_{i=1}^t C_i$ is inconsistent, or, $DT(R \restriction \rho(C_1C_2 \cdots C_t)) \leq d$, since a collision of $t$ pigeons implies that we can recover a solution of $R$ after at most $d$ more queries. Therefore, $\mathcal{F}_h$ is now $t$-witnessing, and so $\psE[\mathcal{F}_h] \leq \varepsilon$.
    We now have \[M = \sum_{i=1}^M \sum_{\ell \in L(T_i)} \psE[C_\ell] = \sum_{i=1}^N \psE[\mathcal{F}_h] \leq \varepsilon N < M,\] a contradiction.
\end{proof}

The previous theorem gives us a powerful method to prove lower bounds against levels of the Pecking Order.
In particular, we will be able to show that the \emph{matching pseudoexpectation} is $t$-collision-free against $\ourPigeon{(t+1)}{M}{N}$ for \emph{all} $t$.
The main observation here is that the union of $t$ matchings from $M$ to $N$ has a collision of size at most $t$, and therefore cannot witness a $(t+1)$-collision of pigeons.
This means that any $t$-subset of conjunctions from a $t$-witnessing family of matchings is inconsistent.
The next technical lemma shows how to upper-bound the weight of such inconsistent families.

\begin{lemma}
    \label{lem:decision-tree-bound}
   Let $x_1, \dots, x_n$ be a set of variables, and let $t, d$ be chosen so that $(t-1)d^2 \leq n$.
   Let $\mathcal{F}$ be any family of degree $\leq d$ conjunctions over these variables, such that for every subset $\mathcal{S} \subseteq \mathcal{F}$ with $|\mathcal{S}| = t$, $\prod_{C \in \mathcal{S}} C \equiv 0$.
   If $\psE$ is a degree $D \geq (t-1)d^2$ pseudoexpectation operator, then $\psE[\mathcal{F}] \leq t-1$.
\end{lemma}
\begin{proof}
    Let us first observe that the statement of the lemma would obviously be true if $\psE$ were the expectation over a \emph{true} probability distribution.
    This is because at most $t-1$ distinct conjunctions in $\mathcal{F}$ are consistent with any total assignment, and thus no set of $t$ conjunctions can be simultaneously activated under a sample from the true probability distribution.

    To prove this claim for $\psE$ we first need to introduce some notation.
    Let $T$ be any decision tree querying the variables $x_1, \ldots, x_n$ and outputting $0$ or $1$. 
    Let $L_b(T)$ denote the leaves of $T$ labelled with $b \in \B$, and let $L(T) = L_0(T) \cup L_1(T)$.
    For any leaf $\ell \in L(T)$ let $C_\ell$ denote the conjunction of literals on the path from the root to $\ell$, and let $\rho_\ell = \rho(C_\ell)$ denote the partial assignment corresponding to this path.
    If the depth of $T$ is at most $d$, define \[ \psE[T] := \sum_{\ell \in L_1(T)} \psE[C_\ell].\]
    An easy induction on the depth of $T$ shows that $\psE[T] \leq 1$.

    Starting with the family $\mathcal{F}$, we create a depth $\leq (t-1)d^2$ decision tree $T$ such that \[\psE[\mathcal{F}] \leq (t-1)\psE[T] \leq t-1.\]
    If $\rho$ is a partial assignment, then let $\mathcal{F} \restriction \rho = \set{C \restriction \rho : C \in \mathcal{F}}$, where it is understood that we remove any conjunctions that are set to $0$ or $1$ under the restriction $\rho$.

    We construct the decision tree $T$ by the following recursive algorithm.
    The decision tree maintains a partial assignment to the above variables. 
    Initially, $\rho = \emptyset$.
    The algorithm proceeds in rounds.
    If $\mathcal{F} \restriction \rho = \emptyset$, we halt and output $1$, if there are any conjunctions in $\mathcal{F}$ consistent with $\rho$, or $0$ otherwise.
    We proceed assuming $\mathcal{F} \restriction \rho \neq \emptyset$.
    In this case, we choose $t-1$ conjunctions $C_1, C_2, \ldots, C_{t-1}$ in $\mathcal{F} \restriction \rho$ and query all unqueried variables among these conjunctions --- if there are less than $t-1$, then we simply query all the variables among all remaining conjunctions.
    After this querying stage, we have learned a partial assignment $\sigma$ to the newly queried variables, and we then recurse on the family $\mathcal{F} \restriction \rho\sigma$.

    The construction above plainly terminates on every branch, since we are reducing the length of each conjunction after every round of querying.
    We argue that the depth of the tree is at most $(t-1)d^2$ and that $\psE[\mathcal{F}] \leq (t-1)\psE[T]$, which completes the proof of the theorem.

    First, we argue that the depth of the tree is at most $D$.
    To see this, we observe that in each round we query at most $(t-1)d$ variables, and we argue that on every branch the algorithm terminates after at most $d$ rounds.
    To see this, consider the $i$-th round, where we query conjunctions $C_1, \dots, C_{t-1}$. 
    Since $\mathcal{F}$ is $t$-witnessing, it follows that every conjunction $C$ remaining in $\mathcal{F}$ must conflict with at least one literal contained in $C_1, \dots, C_{t-1}$.
    Therefore, after querying all unqueried variables in $C_1, \dots, C_{t-1}$, we must query at least one variable from every remaining conjunction in $\mathcal{F}$.
    This means that at the end of the $i$th round, we must reduce the length of every remaining conjunction by at least one. 
    Since each conjunction has at most $d$ variables to begin with, it follows that the entire process can proceed for at most $d$ rounds.
    Thus, the depth of the tree is at most $D$.

    Let us now see that $\psE[\mathcal{F}] \leq (t-1)\psE[T]$.
    First, we observe that since the depth of $T$ is at most $(t-1)d^2$, it follows that every for every leaf $\ell$ of $T$ the conjunction $C_\ell$ has degree at most $D$.
    This means that $\psE[T]$ is well-defined since $\psE$ is a degree-$D$ pseudoexpectation.
    
    So, it remains to show that $\psE[\mathcal{F}] \leq (t-1)\psE[T],$ noting that $\psE[T] \leq 1$ for any decision tree $T$.
    For any node $u$ in $T$ let $T_u$ denote the subtree rooted at $u$, and let $C_u$ denote the conjunction of the literals along the path to $u$. We prove by induction on the height of $u$ that 
    \begin{equation}\label{eq:ind}
        \sum_{C \in \mathcal{F}} \psE[C_uC] \leq (t-1)\sum_{\ell \in L_1(T_u)} \psE[C_uC_\ell].
    \end{equation}
    Once we have this equation, setting $u$ to be the root node $r$ yields \[\sum_{C \in \mathcal{F}} \psE[C] \leq (t-1) \sum_{\ell \in L_1(T)}\psE[C_\ell] = (t-1)\psE[T],\] as desired.

    If $u$ is a $1$-leaf node of $T_u$, then $L_1(T_u) = \set{u}$ and so \[ \sum_{\ell \in L_1(T_u)} \psE[C_uC_\ell] = \psE[C_u].\] On the other hand, for any leaf $u$, by construction of the tree if $C_uC \not \equiv 0$ then $C_uC = C_u$. 
    Since the family $\mathcal{F}$ is $t$-witnessing, there are at most $t-1$ possible choices of $C \in \mathcal{F}$ such that $C_uC \not \equiv 0$ since every set of $t$ conjunctions in $\mathcal{F}$ are inconsistent.
    Therefore \[ \sum_{C \in \mathcal{F}} \psE[C_uC] \leq (t-1)\psE[C_u] = (t-1)\sum_{\ell \in L_1(T_u)} \psE[C_uC_\ell],\] proving the base case of \cref{eq:ind}.
	
    For the inductive step, consider a node $u$ in $T$ querying a variable $x_i$ with children $v_0, v_1$ corresponding to the two outcomes of the query.
    Then
    \begin{align*}
        \sum_{C \in \mathcal{F}} \psE[C_uC] & = \sum_{C \in \mathcal{F}} \psE[C_ux_i C + C_u(1-x_i)C] \\
        & = \sum_{C \in \mathcal{F}} \psE[C_ux_iC] + \psE[C_u(1-x_i)C] \\
        & = \sum_{C \in \mathcal{F}} \psE[C_{v_1}C] + \sum_{C \in \mathcal{F}} \psE[C_{v_0}C] \\
        & \leq (t-1)\sum_{\ell \in L_1(T_{v_1})} \psE[C_{v_1}C_\ell] + (t-1)\sum_{\ell \in L_1(T_{v_0})} \psE[C_{v_0}C_\ell] \\
        & = (t-1)\left(\sum_{\ell \in L_1(T_{v_1})} \psE[C_ux_iC_\ell] + \sum_{\ell \in L_1(T_{v_0})} \psE[C_{u}(1-x_i)C_\ell]\right) \\
        & = (t-1) \sum_{\ell \in L_1(T_u)} \psE[C_uC_{\ell}],
    \end{align*}
    where the inequality follows by the induction hypothesis, and the last equality follows since the leaves of $T_u$ are exactly the union of leaves of $T_{v_0}$ and $T_{v_1}$.
    This proves \cref{eq:ind} for all nodes $u$ of $T$, completing the proof.
\end{proof}

With this technical lemma in hand, we are now ready to prove the main theorem of this section.

\begin{theorem}
\label{thm:lb}
Let $t, d, M, N, M', N'$ be positive integers chosen so that $M \geq tN + 1$, $M' \geq (t-1)N' + 1$, and $(t-1)d^2 \leq N/2$.
Then the $\ourPigeon{(t+1)}{M}{N}$ problem does not have a depth-$d$ decision-tree reduction to $\ourPigeon{t}{M'}{N'}$. 
\end{theorem}
\begin{proof}
    Let $D = (t-1)d^2 \leq N/2$, and let $\psE$ be the degree-$D$ pseudoexpectation obtained from the degree-$D$ matching pseudodistribution (\cref{def:matching-pseudodistribution}) for $\ourPigeon{(t+1)}{M}{N}$.
    We show that $\psE$ is $(d,t,t-1)$-collision-free for $\ourPigeon{(t+1)}{M}{N}$.
    By \cref{thm:psExp}, this implies that $\ourPigeon{(t+1)}{M}{N}$ does not depth-$d$ reduce to $\ourPigeon{t}{M'}{N'}$ for any $M', N'$ with $M' \geq (t-1)N' + 1$.
    
    Consider any $t$-witnessing family $\mathcal{F}$ for $\ourPigeon{(t+1)}{M}{N}$ in which every conjunction has degree $\leq d$.
    Without loss of generality, we can remove any conjunction $C$ from $\mathcal{F}$ such that $\psE[C] = 0$.
    Thus we can assume that $\rho(C)$ encodes a partial matching for all $C \in \mathcal{F}$.
    Let $\mathcal{S} \subseteq \mathcal{F}$ be any collection of $t$ conjunctions in $\mathcal{F}$, let $C_{\mathcal{S}} = \prod_{C \in \mathcal{S}} C$ denote the conjunction obtained by multiplying all conjunctions in $\mathcal{S}$.
    
    Suppose that $C_{\mathcal{S}} \not \equiv 0$, and let $\rho = \rho(C_\mathcal{S})$ denote the corresponding partial assignment.
    Since each constituent conjunction of $C_\mathcal{S}$ is a partial matching, and there are only $t$ conjunctions, it follows that $\rho$ cannot witness a solution to $\ourPigeon{(t+1)}{M}{N}$ since it can only contain a collision of at most $t$ pigeons in any hole.
    A simple adversary strategy then implies that $DT(\ourPigeon{(t+1)}{M}{N} \restriction \rho) = \Omega(N)$, since we can respond to any unqueried pigeon by placing that pigeon into any hole with $\leq t-1$ pigeons. 
    Since the family is $t$-witnessing, it therefore follows that $C_\mathcal{S} \equiv 0$, i.e., $\rho(C_\mathcal{S})$ is an inconsistent partial assignment that tries to place at least one pigeon in two different holes.
    By \cref{lem:decision-tree-bound}, this means that $\psE[\mathcal{F}] \leq t-1$, and therefore $\psE$ is a $(d, t, t-1)$-collision-free pseudoexpectation operator for $\ourPigeon{(t+1)}{M'}{N'}$.
    Applying \cref{thm:psExp} completes the proof.
\end{proof}

Our main result separating the Pecking Order is now an immediate corollary of the previous theorem. We recall it here for convenience.
\pihsep*

Further, we note that we can prove a generalization of \cref{thm:lb}, separating the problem with parameter $a(n)$ from $b(n)$ whenever $a(n)$ is not polynomially close to $b(n)$. We prove this using the same technique as the theorem above, so we only provide a proof sketch here.

\generalLB*

\begin{proof}[Proof Sketch]
    Similar to the proof of \cref{thm:lb}, we construct a $(d,b(n),b(n)-1)$-Collision-free pseudoexpectation operator for $\ourPigeon{a(n)}{N}{M}$ using the Matching pseudodistribution (\cref{def:matching-pseudodistribution}) combined with \cref{lem:decision-tree-bound}. Here, we used that since $a(n)$ is not polynomially close to $b(n)$, no $b(n)$-collection of partial assignments can be witnessing for $\ourPigeon{a(n)}{M}{N}$. Applying \cref{thm:psExp} completes the proof.
\end{proof}

Note that the statement of the above theorem separates all problems which do not have reductions guaranteed by \cref{thm:polynomial-robustness}. Combining the two, we can conclude \cref{thm:posep}.

\posep*

\section{Separations via the Compression Rate}\label{sec:lb:compression}

    In the previous we gave separations in the Pecking Order via collision number. In this section, we present another type of separation within the Pecking Order, which is due to the difference on the compression rate. Note that we are able to rule out \emph{randomized} reductions in the following theorem.

    \pppnpwpp*
    
    \begin{corollary}
        $\PPP^{dt} \not\subseteq n\textrm{-}\PWPP^{dt}$.
    \end{corollary}

     \paragraph{Notation.} We consider $\pigeon$ instances with $N+1$ pigeons and $N$ holes. With a little abuse of notation, we denote a $\pigeon$ instance by a string $\pmap \in [N]^{N+1}$, where $\pmap_i$ is the hole that pigeon $i$ gets mapped to. Assume that the $\ourPigeon{n}{M}{N}$ instance $f(\pmap)$ reduced from $\pmap$ has $M(n')$ pigeons, $N(n')$ holes, and the solution is any $n'$-collision. We have $M' > N' \cdot (n'-1 + c)$ for a constant $c > 0$ from the theorem statement.
    Formally, for any $i \in [M']$, pigeon $i$ from $f(\pmap)$ is mapped to hole $f_i(\pmap)$. Without loss of generality, we assume all decision trees $(f_i, g_o)_{i \in [M'], o \in [M']^{n'}}$ 
    have the same depth $d = \poly(\log N)$.

    We then proceed in four steps:
    \begin{enumerate}
        \item Finding an appropriate distribution $D_N$ of hard $\pigeon$ instances.
        \item Showing that with high probability, there exists a ``non-witnessing'' solution in the $\ourPigeon{n}{M}{N}$ instance $f(\pmap)$ when $\pmap$ is drawn from $D_N$.
        \item Arguing that the error probability is high if all decision trees $(g_o)$ are depth-$0$.
        \item Generalizing the error analysis to depth-$d$.
    \end{enumerate}

    \paragraph{A hard distribution.} By Yao's Minimax principle, it suffices to consider a family of distributions $(D_N)$ of $\pigeon$ instances, and then show that any deterministic low-depth black-box reduction $(f_i, g_o)$ must be wrong with high probability. A natural choice of $D_N$ is taking $\pmap_1, \ldots, \pmap_N$ to be a random permutation over $[N]$, while $\pmap_{N+1}$ is always set to $1$. Now, the only possible solution is the index $i^* \in [N]$ such that $\pmap_{i^*} = 1$. In the rest of this proof, we use $i^*$ rather than the collision pair $(i^*, N+1)$ to indicate the solution of $\pmap$; we also ignore $\pmap_{N+1}$ and assume $\pmap$ is permutation on $[N]$.

    \paragraph{Find a non-witnessing solution.} 

    For any $i \in [M'], \pmap \in D_N$, we say pigeon $i$ in $f(\pmap)$ is \emph{non-witnessing} if the decision tree path in $f_i$ realized by $h$ is not witnessing, i.e., $f_i$ does not query $i^*$ when evaluating on $h$; otherwise, we call pigeon $i$ \emph{witnessing}.
    
    Let $o = (i_1, \ldots, i_{n'})$ be any solution (i.e., an $n'$-collision) of $f(h)$. By taking the union of all decision tree paths in $(f_{i_j})_{j \in [n']}$ realized by $h$, 
    we get a partial assignment $\pmap_o$ of size at most $d \cdot n' = \poly(\log N)$ in $\pmap$. 
    Without loss of generality, we assume this partial assignment $\pmap_o$ is also returned by the reduction as part of the solution.
    
    We say $o$ is a \emph{non-witnessing} solution if for all $j \in [n']$, the pigeon $i_j$ is non-witnessing, i.e., the partial assignment $\pmap_o$ does not witness the location of $i^*$; otherwise, $o$ is a \emph{witnessing} solution.
    Intuitively, a non-witnessing solution reveals almost no information about the key location $i^*$.
    
    We now prove the following key lemma.


    \begin{lemma}\label{lem:claim1}
        When $\pmap$ is drawn from $D_N$, $f(\pmap)$ has a non-witnessing solution with probability $1 - \negl(\log N)$.
    \end{lemma}

    \begin{proof}
        Since $\pmap$ is a random permutation, and $f_i$ only has $d$ levels, for any pigeon $i \in [M']$ in $f(\pmap)$, we have 

        $$\Pr_{h \sim D_N}[i \textrm{ is witnessing}] \leq \frac{d}{N} = \negl(\log N).$$
        
        Using Markov's inequality, the probability that $f(\pmap)$ has at least $cN'$ witnessing pigeons is $\negl(\log N)$. 
        In other words, with probability $1 - \negl(\log N)$, $f(\pmap)$ has more than $(n'-1) \cdot N'$ non-witnessing pigeon.
        
        Therefore, with probability $1 - \negl(\log N)$, $f(\pmap)$ has a $n'$-collision with only non-witnessing pigeons, i.e., a non-witnessing solution.
        
    \end{proof}
    
    \paragraph{Success probability for depth-$0$.}

    We say a reduction is depth-$k$ ($k \leq d$) if all the decision trees $(g_o)$ are depth-$k$, while $(f_i)$ are still depth-$d$. For a fixed family of depth-$d$ decision trees $(f_i)$, define $\ps_k$ as the maximal success probability of any depth-$k$ ($k \leq d$) reduction.
    We first consider the success probability of depth-$0$ reduction, i.e., $g_o \in [N]$ is a fixed location. 

    Let $R, Q$ be the short-hand for \pigeon and $\ourPigeon{n}{M}{N}$.
    Formally, our goal in this step is to show the following inequality.

    
    \begin{lemma}\label{lem:depth_0}
    $$\ps_0 \coloneqq \max_{f,g} \Pr_{\pmap \sim D_N} [(\pmap, g_o) \in R  \Leftarrow (f(\pmap), o) \in Q] < \negl(\log N).$$
    \end{lemma}
    Let $\pmap^{(1)}$ be an instance of \pigeon with solution $i^*$. The key technical trick here is to roll a second dice, which helps us estimate the error probability. Specifically, we take a uniformly random index $j^* \in [N]$, and then generate $\pmap^{(2)}$ from $\pmap^{(1)}$ by swapping the solution from location $i^*$ to $j^*$. Formally, $$\pmap^{(2)}_{j^*} = 1, \pmap^{(2)}_{i^*} = \pmap^{(1)}_{j^*}; \;\; \pmap^{(2)}_i = \pmap^{(1)}_i, \forall i \in [N] \backslash \{i^*, j^*\}.$$ 
    By symmetry, we know the marginal distribution of $\pmap^{(2)}$ is also a random permutation, so we can calculate the success probability on $\pmap^{(2)}$ instead.

    Let $\uEvent(\pmap, o)$ indicate the event that $o$ is a non-witnessing solution for $f(\pmap)$, and denote the uniform distribution over $[N]$ by $U_N$. 
    We have the following observation regarding our second dice. 

    \begin{lemma}\label{lemma:useless_h1h2}
    For any $\pmap^{(1)}$ and $o$,
    $$\Pr_{j^* \sim U_N} [\uEvent(\pmap^{(2)}, o) \,|\, \uEvent(\pmap^{(1)}, o)] = 1 - \negl(\log N).$$
    \end{lemma}

    \begin{proof}
         Since $o$ is non-witnessing, the partial assignment $\pmap_o$ does not contain location $i^*$. Recall that $\pmap^{(2)}$ is different to $\pmap^{(1)}$ only on location $i^*$ and $j^*$. Therefore, with $1 - \frac{|h_o|}{n} = 1 - \negl(\log N)$ probability on $j^*$, $\pmap_o$ is also a partial assignment of $\pmap^{(2)}$, which implies that $o$ is a non-witnessing solution of $f(\pmap^{(2)})$.
    \end{proof}
    
    Now we are ready to prove \cref{lem:depth_0}.
    
    \begin{proof}[Proof of \cref{lem:depth_0}]

    Fix an arbitrary depth-$0$ reduction $(f,g)$.
    We give an overview before diving into the calculations.
    We first randomly draw $\pmap^{(1)}$ from $D_N$, and there are two possible cases: either $f(\pmap^{(1)})$ has a non-witnessing solution, or $f(\pmap^{(1)})$ does not have any non-witnessing solution. The second case will happen with very low probability (\cref{lem:claim1}), so we can assume $f(\pmap^{(1)})$ has a non-witnessing solution $o^*$.
    We then roll a second dice $j \sim [N]$ and generate $\pmap^{(2)}$ accordingly. We know that $o^*$ is also a non-witnessing solution of $f(\pmap^{(2)})$ with high probability (\cref{lemma:useless_h1h2}); among these $\pmap^{(2)}$, the reduction could possibly be correct only when $j = g_{o^*}$. 

    Formally, let $o^*$ be the first non-witnessing solution of $f(\pmap^{(1)})$ in the lexicographical order, if $f(\pmap^{(1)})$ has a non-witnessing solution; otherwise, let $o^*$ to be the lexicographically first solution of $f(\pmap^{(1)})$.
    We first consider whether $\pmap^{(1)}$ has a non-witnessing solution.
    
    \begin{align}
       & \Pr_{\pmap^{(2)} \sim D_N} [(\pmap^{(2)}, g_o) \in R  \Leftarrow (f(\pmap^{(2)}), o) \in Q] \notag\\
        &= \Pr_{\pmap^{(1)} \sim D_N, j^* \sim U_N} [(\pmap^{(2)}, g_o) \in R  \Leftarrow (f(\pmap^{(2)}), o) \in Q] \notag\\
        &\leq \Pr_{\pmap^{(1)}, j^*} [(\pmap^{(2)}, g_o) \in R  \Leftarrow (f(\pmap^{(2)}), o) \in Q \,|\, \uEvent(\pmap^{(1)}, o^*)]  + \Pr_{\pmap^{(1)}}[\neg \uEvent(\pmap^{(1)}, o^*)] \label{equ:use_claim1} 
    \end{align}

    Note that the term $\Pr_{\pmap^{(1)}}[\neg \uEvent(\pmap^{(1)}, o^*)]$ is at most $\negl(\log N)$ from \cref{lem:claim1}. So it remains to bound the first item in inequality~\cref{equ:use_claim1}, which is the success probability on $\pmap^{(2)}$ condition on $\pmap^{(1)}$ has a non-witnessing solution $o^*$.

    \begin{align}
        & \Pr_{\pmap^{(1)}, j^*} [(\pmap^{(2)}, g_o) \in R  \Leftarrow (f(\pmap^{(2)}), o) \in Q \,|\, \uEvent(\pmap^{(1)}, o^*)] \notag\\
        &\leq \Pr_{\pmap^{(1)}, j^*} [(\pmap^{(2)}, g_{o^*}) \in R  \wedge  \uEvent(\pmap^{(2)}, o^*) \,|\, \uEvent(\pmap^{(1)}, o^*)] \notag \\ 
        &\quad + \Pr_{\pmap^{(1)}, j^*} [ \neg \uEvent(\pmap^{(2)}, o^*) \,|\, \uEvent(\pmap^{(1)}, o^*)] \notag\\
        &\leq \Pr_{\pmap^{(1)}, j^*} [(\pmap^{(2)}, g_{o^*}) \in R  \wedge  \uEvent(\pmap^{(2)}, o^*) \,|\, \uEvent(\pmap^{(1)}, o^*)] + \negl(\log N) \label{equ:use_obs}\\
        &= \Pr_{\pmap^{(1)}, j^*} [(j^* =  g_{o^*})  \wedge  \uEvent(\pmap^{(2)}, o^*) \,|\, \uEvent(\pmap^{(1)}, o^*)]+ \negl(\log N) \notag\\
        &\leq \Pr_{\pmap^{(1)}, j^*} [(j^* =  g_{o^*}) \,|\, \uEvent(\pmap^{(1)}, o^*)]+ \negl(\log N) \notag\\
        &= \frac{1}{N} + \negl(\log N) \notag\\
        &= \negl(\log N). \label{equ:item1}
    \end{align}

    Note that inequality~\cref{equ:use_obs} comes from \cref{lemma:useless_h1h2}. 

    \end{proof}

    \paragraph{Success probability for depth-$d$.}

    We now consider depth-$k$ reduction for $k \leq d$, i.e., all $(g_o)$ are depth-$k$.
    Define event $\hitI(\pmap^{(1)}, o)$ to be true if $o$ is a witnessing solution of $f(\pmap^{(1)})$, or $i^*$ is queried when evaluating $g_o$ on $\pmap^{(1)}$. Intuitively, $g_o(\pmap^{(1)})$ has to \emph{guess} a location if $\hitI(\pmap^{(1)}, o)$ is not true,

    The following lemma effectively reduces the problem to the case of depth-$(k-1)$.
    \begin{lemma}\label{lem:hit}
        For a depth-$k$ reduction $(f_i, g_o)$,
        $$\Pr_{\pmap^{(1)} \sim D_N} [\hitI(\pmap^{(1)}, o) \Leftarrow (f(\pmap^{(1)}), o) \in Q] \leq \ps_{k-1}$$
    \end{lemma}

    \begin{proof}
        We construct a depth-$(k-1)$ reduction $(f_i, g'_o)$: For each possible $o \in [M']^{n'}$, the structure of $g'_o$ is same as the first $k-1$ levels of $g_o$. If the solution $i^*$ is witnessed by $o$ itself or the first $k-1$ queries of $g'_o$, then $g'_o$ will output $i^*$; otherwise, $g'_o$ will output the location that is going be queried in $g_o$ in the $k$-th level, if $g_o$ is evaluated on the same input. Therefore, we have
        $$\Pr_{\pmap^{(1)} \sim D_N} [\hitI(\pmap^{(1)}, o) \Leftarrow (f(\pmap^{(1)}), o) \in Q] = \Pr_{\pmap^{(1)} \sim D_N} [(\pmap^{(1)}, g'_o) \in R \Leftarrow (f(\pmap^{(1)}), o) \in Q] \leq \ps_{k-1}.$$
    \end{proof}

    Define event $\hitJ(\pmap^{(1)}, o, j^*)$ to be true if the partial assignment $\pmap_o$ contains the location $j^*$, or $j^*$ is queried when $g_o$ is evaluated on $\pmap^{(1)}$. Note that when both $\hitI(\pmap^{(1)}, o)$ and $\hitJ(\pmap^{(1)}, o, j^*)$ are false, we have $g_o(\pmap^{(1)}) = g_o(\pmap^{(2)})$, because they only differ in location $i^*$ and $j^*$.
    We have the following lemma using the same argument in \cref{lemma:useless_h1h2}.
    \begin{lemma}\label{lemma:nothit_h1h2}
    For any $\pmap^{(1)}$ and $o$,
    $$\Pr_{j^* \sim U_N} [\neg \hitJ(\pmap^{(1)}, o, j^*) \,|\,\neg \hitI(\pmap^{(1)}, o)] = 1 - \negl(\log N).$$
    \end{lemma}

    Now, let us wrap everything up.
    \begin{proof}[Proof of \cref{thm:lb:compression}]
        We follow a similar strategy as in the proof of \cref{lem:depth_0}. For any depth-$k$ reduction $(f,g)$, there are two possibilities regarding to $\pmap^{(1)}$: either $f(\pmap^{(1)})$ has a (non-witnessing) solution $o^*$ that $\hitI(\pmap^{(1)}, o^*)$ is false, or $\hitI(\pmap^{(1)}, o)$ is true for any solution $o$ of $f(\pmap^{(1)})$. The second case will happen with probability at most $\ps_k$ by \cref{lem:hit}, and we also roll a second dice $j^*$ to analyze the first case.

        Formally, let $o^*$ to be the lexicographically first solution of $f(\pmap^{(1)})$ that $\hitI(\pmap^{(1)}, o)$ is false, if such a solution exists; otherwise, let $o^*$ to be the lexicographically first solution of $f(\pmap^{(1)})$. We have

        \begin{align}
        &  \Pr_{\pmap^{(1)} \sim D_N, j^* \sim U_N} [(\pmap^{(2)}, g_o(\pmap^{(2)})) \in R  \Leftarrow (f(\pmap^{(2)}), o) \in Q] \notag\\
        &\leq \Pr_{\pmap^{(1)}, j^*} [(\pmap^{(2)}, g_o(\pmap^{(2)})) \in R  \Leftarrow (f(\pmap^{(2)}), o) \in Q \,|\, \neg \hitI(\pmap^{(1)}, o^*)] + \Pr_{\pmap^{(1)}}[\hitI(\pmap^{(1)}, o^*)] \notag\\ 
        &\leq \Pr_{\pmap^{(1)}, j^*} [(\pmap^{(2)}, g_o(\pmap^{(2)})) \in R  \Leftarrow (f(\pmap^{(2)}), o) \in Q \,|\, \neg \hitI(\pmap^{(1)}, o^*)] + \ps_{k-1} \label{equ:use_hit}\\ 
        &\leq \Pr_{\pmap^{(1)}, j^*} [(j^* =  g_{o^*}(\pmap^{(2)}))  \wedge  \neg \hitJ(\pmap^{(1)}, o^*, j^*) \,|\, \neg \hitI(\pmap^{(1)}, o^*)] \notag \\ 
        &\quad + \Pr_{\pmap^{(1)}, j^*} [ \hitJ(\pmap^{(1)}, o^*, j^*) \,|\, \neg \hitI(\pmap^{(1)}, o^*)] + \ps_{k-1} \notag\\
        &\leq \Pr_{\pmap^{(1)}, j^*} [(j^* =  g_{o^*}(\pmap^{(2)}))  \wedge  \neg \hitJ(\pmap^{(1)}, o^*, j^*) \,|\, \neg \hitI(\pmap^{(1)}, o^*)] + \ps_{k-1} + \negl(\log N) \label{equ:use_obs2}\\
        &\leq \Pr_{\pmap^{(1)}, j^*} [(j^* =  g_{o^*}(\pmap^{(1)})) \,|\, \neg \hitI(\pmap^{(1)}, o^*)]+ \ps_{k-1} + \negl(\log N) \notag\\
        &= \frac{1}{N} + \ps_{k-1} + \negl(\log N) \notag\\
        &= \ps_{k-1} + \negl(\log N). \label{equ:pk}
    \end{align}

    Inequality~\cref{equ:use_hit} uses \cref{lem:hit}, and inequality~\cref{equ:use_obs2} comes from \cref{lemma:nothit_h1h2}. From inequality~\cref{equ:pk}, we have the success probability of any depth-$d$ reduction 
    \begin{equation}\label{equ:final}
    p_d \leq d \cdot \negl(\log N) + p_0 = \negl(\log N).
    \end{equation}
    
    \end{proof}
    


\section{Relationship to Other Classes}
\label{sec:other-classes}

\subsection{Polynomial Local Search ($\PLS$) and Polynomial Parity Argument ($\PPA$)}

Prior work in the literature has shown that $\PPA^{dt} \not \subseteq \PPP^{dt}$ \cite{Beame95} and $\PLS^{dt} \not \subseteq \PPP^{dt}$ \cite{Goeoes2022}.
In this section, we prove that neither of these classes are contained in $\PAP^{dt}$, recalled here.

\plspap*

Both of these separations use the shared concept of \emph{gluability}, which we first define.
\begin{definition}
    Let $C$ be a conjunction and let $T$ be a decision tree.
    A \emph{completion} of $C$ by $T$ is any conjunction of the form $CC_\ell$ where $\ell$ is any leaf of $T$.
\end{definition}

\begin{definition}
    A query total search problem $R \subseteq \B^n \times O$ is \emph{$(d,t)$-gluable} if for every conjunction $C$ of degree at most $d$ there is a depth $O(d)$ decision tree $T_C$ such that the following holds.
    Let $C_1, C_2, \ldots, C_t$ be any sequence of $t$ consistent conjunctions, and let $C_1', \ldots, C_t'$ be any sequence of $t$ conjunctions chosen so that $C_i'$ is a completion of $C_i$ by $T_{C_i}$.
    If $C_i'$ is non-witnessing for each $i = 1, 2, \ldots, t$ and $C_1'C_2'\cdots C_t'$ is consistent, then $DT(R \restriction \rho(C_1'C_2' \cdots C_t')) \geq d$.
\end{definition}

A weaker notion of gluability was introduced by G\"{o}\"{o}s et al. \cite{Goeoes2022} in the case where $t = 2$, in order to prove $\PLS^{dt} \not \subseteq \PPP^{dt}$, although the idea also implicitly appears in \cite{Beame95}.
We show a new result for gluability related to collision-freedom: if a problem $R$ is gluable \emph{and} it admits a non-witnessing pseudoexpectation operator $\psE$, then one can show automatically that $\psE$ is also collision-free, and \cref{thm:psExp} could be applied.
We prove this now in a strong form.

\begin{definition}
    Let $R \subseteq \B^n \times O$ be a query total search problem, and let $\psE$ be a degree-$d$ pseudoexpectation operator.
    Then $\psE$ is \emph{$\varepsilon$-nonwitnessing for $R$} if the following property holds:
    \begin{itemize}
        \item {\bf $\varepsilon$-Nonwitnessing.} $\psE[C] \leq \varepsilon$ for any degree-$d$ conjunction $C$ witnessing $R$.
    \end{itemize}
    Note that $0$-nonwitnessing is synonymous with $R$-nonwitnessing (cf.~\cref{def:pseudoexpectation}).
\end{definition}

Our goal now is to prove the following theorem.

\begin{theorem}
    \label{thm:gluable-implies-pap}
    Let $R = \set{R_n \subseteq \B^n \times O_n}$ be a query total search problem not contained in $\PPADS^{dt}$, and suppose $t = O(\poly(\log n))$.
    If $R$ is $(p(\log(n)), t)$-gluable for any polynomial function $p$, then $R \not \in t\text{-}\PPP^{dt}$.
\end{theorem}

To prove this theorem, we will need to define the \emph{Sherali-Adams} proof system.
To avoid introducing proof-complexity preliminaries, we will state the definition of the proof system directly in terms of the total search problems $R$, although we refer the interested reader to \cite{Goeoes2022} for technical details.
If $R \subseteq \B^n \times O$ is a query total search problem in $\TFNP$, then we define a related unsatisfiable CNF formula $\cnf(R)$ that encodes the (false) statement ``$R$ is not total''.
Formally, \[\cnf(R) := \bigwedge_{o \in O} \bigwedge_{\ell \in L_1(T_o)} \neg C_{\ell}\]
where $\set{T_o}_{o \in O}$ is the family of $\poly(\log(n))$-depth decision trees witnessing solutions of $R$, and $L_1(T_o)$ is the set of $1$-leaves of the decision tree $T_o$.
We will be interested in \emph{refutations} of the formula $\cnf(R)$ --- in other words, proofs of the tautology ``$R$ is total''.
Recall that a \emph{conical junta} is any non-negative linear combination of conjunctions, that is an expression of the form $J = \sum_D \lambda_D D$, where each $D$ is a conjunction and $\lambda_D > 0$.
The \emph{degree} of $J$ is the maximum degree of any conjunction $D$ appearing in the sum.
The \emph{magnitude} of $J$, denoted $||J||$, is $\max_D \lambda_D$.

\begin{definition}
    Let $R \in \TFNP^{dt}$ be a total query search problem.
    A \emph{$\ZZ$-Sherali-Adams proof of totality for $R$} is a sequence of conical juntas $\Pi = (J, J_\ell)_{o \in O, \ell \in L_1(T_o)}$ over the variables of $R$ such that
    \[-1 = \sum_{o \in O} \sum_{\ell \in L_1(T_o)} -J_\ell C_\ell + J \]
    where we are working in multilinear polynomial algebra.
    The \emph{degree} of the refutation $\Pi$ is $\deg(\Pi):= \max \set{\deg J_o C_o}_{o} \cup \set{\deg J}.$
    The \emph{magnitude} of the refutation $\Pi$ is $||\Pi|| := \max \set{||J_o||}_{o} \cup \set{J}$. 
\end{definition}

G\"{o}\"{o}s et al. \cite{Goeoes2022} proved that that $R \in \PPADS^{dt}$ if and only if it admits Sherali-Adams proofs with low degree and magnitude.

\begin{theorem}[\cite{Goeoes2022}]
    \label{thm:sa-theorem}
   Let $R \in \TFNP^{dt}$ be a total query search problem.
   Then $R \in \PPADS^{dt}$ if and only if for some constant $c$, $R$ admits a $\log^c(n)$-degree, $n^{\log^c(n)}$-magnitude $\ZZ$-Sherali-Adams proof of totality.
\end{theorem}

One can show that \emph{any} lower bound against degree-$d$ Sherali-Adams with small magnitude implies the existence of a weak pseudoexpectation operator (in other words, $\varepsilon$-nonwitnessing operators are \emph{complete} for unary Sherali-Adams lower bounds).
The proof of this follows the usual proof of completeness for Sherali-Adams proofs via convex duality (see e.g.~\cite{Fleming19}), and was first observed by Hub\'{a}\v{c}ek, Khaniki, and Thapen \cite{HubacekKT24}.
\begin{theorem}[\cite{HubacekKT24}]
    \label{thm:HKT24}
    If there is no degree-$d$ $\ZZ$-Sherali-Adams proof of totality for $R$ of magnitude $\leq k$, then there is a degree-$d$, $1/k$-nonwitnessing pseudoexpectation operator for $R$.
\end{theorem}

We can now prove \cref{thm:gluable-implies-pap}.

\begin{proof}[Proof of \cref{thm:gluable-implies-pap}]
    Let $R \subseteq \B^n \times O$ be any $(p(\log n), t)$-gluable total search problem not contained in $\PPADS$.
    Suppose by way of contradiction that $R \in t\text{-}\PPP^{dt}$, and so $R$ has a depth-$d$ reduction to $\ourPigeon{t}{(t-1)N+1}{N}$ where $d = \log^{C_0} n$ for some universal constant $C_0$ and $N = n^{\log^{C_0} n}$.
    We use the fact that $R$ is $(d^*, t)$-gluable for $d^* = \log^{2C_0} n = \omega(d)$.
    Since $R \not \in \PPADS^{dt}$, it follows from \cref{thm:sa-theorem} that for \emph{every} positive constant $\alpha$, there is no degree-$\log^{\alpha} n$, $\ZZ$-Sherali-Adams proof of totality for $R$ with magnitude $n^{\log^{\alpha} n}$.
    By \cref{thm:HKT24}, the non-existence of a $\ZZ$-Sherali-Adams proof implies that there is a degree-$\log^{\alpha} n$, $1/n^{\log^{\alpha}n}$-nonwitnessing pseudoexpectation operator $\psE$ for $R$.
    With this in mind, let $\psE$ be a degree-$t\log^{4C_0} n$, $\varepsilon$-nonwitnessing pseudoexpectation operator for $R$ with $\varepsilon = 1/n^{t\log^{4C_0} n}$.
    Our goal is to show that $\psE$ is $(d, t, t-1 + o(1/N))$-collision-free, which will contradict the existence of the assumed reduction to $\tightPigeon{t}$ by \cref{thm:psExp}.
    
    Let $\mathcal{F}$ be any $t$-witnessing family of width $\leq d$ conjunctions, and let $\mathcal{C}(\mathcal{F})$ be the family obtained by replacing each conjunction $C$ in $\mathcal{F}$ with all of its completions $CC_\ell$ for $\ell \in L(T_{C})$ guaranteed by $(d^*, t)$-gluability.
    We first observe that $\psE[\mathcal{C}(\mathcal{F})] = \psE[\mathcal{F}]$, noting that the degree of $\psE$ is $\omega(d)$ and therefore both of these expressions are well-defined.
    To see this, consider any conjunction $C \in \mathcal{F}$.
    If $T_C$ is the decision tree completing $C$ in the definition of gluability, then an easy induction on the depth of the decision tree shows that \[1 = \sum_{\ell \in L(T_C)} C_\ell\] where the equality is between multilinear polynomials.
    This implies that $C = \sum_{\ell \in L(T_C)} CC_\ell$ and thus \[ \psE[C] = \sum_{\ell \in L(T_C)} \psE[CC_\ell].\]
    It immediately follows that $\psE[\mathcal{F}] = \psE[\mathcal{C}(F)]$, since $\mathcal{C}(\mathcal{F})$ is obtained by replacing each conjunction $C$ with its completions.
    This means that it suffices to bound the weight of $\mathcal{C}(\mathcal{F})$ instead of $\mathcal{F}$.
    Let $W \subseteq \mathcal{C}(\mathcal{F})$ be the collection of all conjunctions in $\mathcal{C}(\mathcal{F})$ that are themselves witnessing, and let $X = \mathcal{C}(\mathcal{F}) \setminus W$ be the remaining conjunctions, and note that $\psE[\mathcal{C}(\mathcal{F})] = \psE[X] + \psE[W]$. 
    
    First, let's bound the weight of $\psE[W]$.
    Since every conjunction $C' \in W$ is witnessing, it follows that $\psE[C'] \leq \varepsilon$ since $\psE$ is $\varepsilon$-nonwitnessing for $R$.
    This means that $\psE[W] \leq |W|\varepsilon \leq \varepsilon n^{O(d)}$, since every conjunction in $W$ has width at most $O(d)$.

    We now bound the weight of $X$.
    Let $C_1, C_2, \ldots, C_t$ be any sequence of distinct conjunctions chosen from $X$.
    Since $\mathcal{F}$ is $t$-witnessing, $\mathcal{C}(\mathcal{F})$ is also $t$-witnessing, and this implies that if $C' = C_1C_2 \ldots C_t$ is consistent, then $DT(R \restriction \rho(C')) = O(d)$.
    However, by the definition of $(d^*, t)$-gluability, whenever $C'$ is consistent and each $C_i$ is non-witnessing, then $DT(R \restriction \rho(C')) \geq d^*$.
    Since $d^* = \omega(d)$, these two facts together imply that if $C_1, C_2, \ldots, C_t \in X$ are distinct conjunctions, then $C_1C_2 \cdots C_t$ must be inconsistent, since no individual conjunction $X$ can be witnessing.
    Therefore, all such $C'$ composed of conjunctions from $X$ are inconsistent, and so applying \cref{lem:decision-tree-bound} we observe that $\psE[X] \leq t-1$.
    
    Combining the two weight bounds yields $\psE[\mathcal{F}] = \psE[\mathcal{C}(\mathcal{F})] = \psE[W] + \psE[X] = t-1 + \varepsilon n^{O(d)}$.
    Our choice of parameters then implies that \[\varepsilon n^{O(d)} = \frac{n^{O(\log^{C_0} n)}}{n^{t\log^{4C_0} n}} = o\left(\frac{1}{N}\right),\]
    where the last equality follows since $N = n^{\log^{C_0} n}$.
    This means that the pseudoexpectation $\psE$ for $R$ is a $(d, t, t-1 + o(1/N))$-collision-free pseudoexpectation operator.
    This is a contradiction to \cref{thm:psExp}, and it follows that $R \not \in t\text{-}\PPP^{dt}$.
\end{proof}

In the remainder of this section, we prove the required gluability results for $\PPA$ and $\PLS$.
\paragraph{$\PPA^{dt}$ is gluable.}

We first introduce the defining problem for the class $\PPA^{dt}$, called $\leaf$.
\begin{definition}[$\leaf_n$]
    This problem is defined on a set of $n$ nodes, denoted by $[n]$, where the node $1$ is ``distinguished''.
    For input, we are given a neighbourhood $N_u \subseteq [n]$ of size $|N(u)| \leq 2$ for each node $u \in [n]$.
    Given this list of neighbourhoods, we create an undirected graph $G$ where we add an edge $uv$ if and only if $u \in N(v)$ and $v \in N(u)$. We say $u \in [n]$ is a \emph{leaf} if it has in-degree $1$ and out-degree $0$.
    The goal of the search problem is to output either
    \begin{enumerate}
	   \item $1$, if $1$ is not a leaf in $G$, or \hfill \emph{(no distinguished leaf)}
	   \item $u \neq 1$, if $u$ is a leaf in $G$. \hfill \emph{(proper leaf)}
    \end{enumerate}
    The class $\PPA^{dt}$ contains all query total search problems with $\poly(\log(n))$-complexity reductions to $\leaf$.
\end{definition}

The seminal work by Beame et al. \cite{Beame95} proved that $\leaf \not \in \PPP^{dt}$, which implies that $\leaf \not \in \PPADS^{dt}$.
Therefore, by \cref{thm:gluable-implies-pap}, to prove $\PPA^{dt} \not \subseteq \PAP^{dt}$, we need to show that $\leaf_n$ is $(\poly(\log n), \poly(\log n))$-gluable.

\begin{lemma}
    $\leaf_n$ is $(p(\log(n)), p(\log(n)))$-gluable for any polynomial $p$.
\end{lemma}
\begin{proof}
    Let $d, t = \poly(\log(n))$.
    Let $C$ be any conjunction of degree $d$ over the variables of $\leaf_n$.
    The decision tree $T_C$ completing $C$ does the following: if $C$ ever queries a node $u \in [n]$, receiving a neighbourhood $N(u)$, $T_C$ queries the nodes in $N(u)$ as well.
    In total, this requires $O(d)$ more queries.
    Now, let $C_1, C_2, \ldots, C_t$ be any sequence of $t$ consistent conjunctions, and let $C_1', C_2', \ldots, C_t'$ be any sequence of $t$ conjunctions chosen so that $C_i'$ is a completion of $C_i$ by $T_{C_i}$.
    Assume that these completions are all consistent and non-witnessing, and let $C' = C_1'C_2' \cdots C_t'$.
    First suppose by way of contradiction that $C'$ witnesses a $\leaf_n$ solution $u$.
    But in either case, the fact that $u$ is a solution must have been witnessed by some $C_i'$ in the sequence, since we have explicitly queried all the neighbourhoods of nodes in each $C_i$.
    This is a contradiction, and thus $C'$ must be non-witnessing.

    Let us now show $\leaf_n \restriction \rho(C')$ has large decision tree depth.
    We can give a simple adversary argument as follows. 
    Let $U$ be the set of all nodes currently queried --- initially, $U$ contains all nodes queried by $C'$, and thus $|U| = \poly(\log(n))$.
    Say a node $u$ is on the \emph{boundary} of $U$ if it has not been queried, but appears as a neighbour of a queried node.
    Consider a decision tree $T$ querying nodes in $\leaf_n$.
    If $T$ queries a node $v$ not in the boundary of $U$, then output $N(v) = \emptyset$.
    Otherwise, suppose $T$ queries a node $v$ in the boundary of $U$.
    If two nodes $\set{u_1, u_2} \subseteq U$ have $v$ in their neighbour set, then set $N(v) = \set{u_1, u_2}$.
    Otherwise, if just one node $u \in U$ has $v$ in its neighbour set, then set $N(v) = \set{u, w}$, where $w$ is any node not in $U$ nor the boundary of $U$.
    Since the degree of every node is $\leq 2$ and $|U| = \poly(\log(n))$, we can clearly continue this adversary strategy for $\Omega(n)$ queries.
    Therefore $DT(\leaf_n \restriction \rho(C')) = \Omega(n)$, and thus $\leaf_n$ is $(\poly(\log(n)), \poly(\log(n)))$-gluable.
\end{proof}

\begin{corollary}
    $\PPA^{dt} \not \subseteq \PAP^{dt}$.
\end{corollary}

\paragraph{$\PLS^{dt}$ is gluable.}
Let us first recall the $\sodLong$ problem, also denoted $\sod$, which is the defining problem for $\PLS$.
Our definition follows that of G\"{o}\"{o}s et al. \cite{Goeoes2022}. 

\begin{definition}[$\sodLong$]
The $\sod_n$ problem is defined on the $[n] \times [n]$ grid, where the node $(1,1)$ is ``distinguished''.
As input, for each grid node $u=(i,j) \in [n] \times [n]$, we are given an ``active bit" $a_u \in \B$. Furhter, we are given a \emph{successor} $s_u \in [n]$, interpreted as naming a node $(i+1, s_u)$ on the next row.
We say a node $u$ is \emph{active} if $a_u = 1$, otherwise it is \emph{inactive}. A node $u$ is a \emph{proper sink} if $u$ is inactive but some active node has $u$ as a successor.
The goal of the search problem is to output either
\begin{enumerate}
	\item \label{sod1} $(1,1)$, if $(1,1)$ is inactive \hfill \emph{(inactive distinguished source)}
	\item \label{sod2} $(n, j)$, if $(n,j)$ is active, \hfill \emph{(active sink)}
	\item \label{sod3} $(i, j)$ for $i \leq n-1$, if $(i,j)$ is active and its successor is a proper sink. \hfill \emph{(proper sink)}
\end{enumerate}
We define the class $\PLS^{dt}$ to be all total query search problems with $\poly(\log(n))$-complexity $\sodLong$-formulations.
\end{definition}

One of the main results of G\"{o}\"{o}s et al. is the following.

\begin{theorem}[\cite{Goeoes2022}, Corollary 1.]
   $\PLS^{dt} \not \subseteq \PPADS^{dt}$. 
\end{theorem}

G\"{o}\"{o}s et al. already showed that the $\sod_n$ problem is $(\poly(\log(n)), 2)$-gluable, but it is not hard to see that their proof generalizes to show that it is actually $(\poly(\log(n)), t)$-gluable for any $t = \poly(\log n)$.
The proof of the following lemma closely follows the argument of \cite[Lemma 13]{Goeoes2022}.

\begin{lemma}
\label{lem:glueable}
    $\sod$ is $(p(\log(n)), p(\log(n)))$-gluable for any polynomial $p$.
\end{lemma}

\begin{proof}
Let $C$ be any conjunction of degree $d = \poly(\log(n))$, defined over the variables of $\sod_n$.
The decision tree $T_C$ completing $C$ starts by checking whether $C$ queries any active node below row $n-d-2$. 
If yes, $T$ picks any one such active node and follows the successor path until a sink is found, making the completion witnessing. Note that this step incurs at most $O(d)$ queries. Finally, $T$ ensures that any query to a successor variable in $C$ is followed by a query to the active bit of the successor. This also costs $O(d)$ further queries.
Thus the depth of $T_C$ is $O(d) = \poly(\log(n))$.

As in the definition of gluability, let $C_1, C_2, \ldots, C_t$ be any sequence of $t$ consistent conjunctions, and let $C_1', C_2', \ldots, C_t'$ be any sequence of completions of those conjunctions by their respective decision trees.
Suppose that $C_i'$ is non-witnessing for each $i$ and that $C_1',C_2',\dots,C_t'$ are consistent, and suppose by way of contradiction that the conjunction $C' = C_1'C_2'\cdots C_t'$ is witnessing.
If it reveals a $\sod$ solution $u$ of type (\hyperref[sod1]{1}) or (\hyperref[sod2]{2}), then it must be that for some $i$, $C_i'$ queries an active bit of $u$: a contradiction with the fact that $C_i'$ is non-witnessing. 
On the other hand, if $C'$ reveals a solution $u$ of type (\hyperref[sod3]{3}), then it must be that for some $i$, $C_i'$ checks for the successor $s_u$ of $u$, but the completion $T_{C_i}$ forces this check to be followed by a query to the active bit of $s_u$, making one of the initial partial assignments witnessing as well. 
Hence $C'$ is non-witnessing.

We finally argue that $R_n\restriction \rho(C')$ has query complexity greater than $d$ by describing an adversary that can fool any further $\ell$ queries to $p$ without witnessing a solution. Recall that $p$ makes no queries to nodes below row $n-d-2$. The adversary answers queries as follows. If the successor pointer of an active node is queried, then we answer with a pointer to any unqueried node on the next row and make it active (there always exists one as $d \ll n$). If a node $u$ is queried that is not the successor of any node, we make $u$ inactive ($a_u = 0$ and $s_u$ is arbitrary). This ensures that a solution can only lie on the very last row $n$, which is not reachable in $\ell$ queries starting from row $n - d - 2$.
\end{proof}

We can thus infer \cref{thm:plspap}, which we restate here for convenience.

\plspap*

\subsection{Iterated Pigeonhole}

In this section, we discuss the relationship between the Pecking Order and classes \PLC and \UPLC defined by \cite{Pasarkar2023}. We give an alternate definition of \UPLC which is more convenient for our reductions. The class \UPLC is defined by its complete problem with the same name \uplc, specified in below.


\begin{definition}[\uplc]\label{def:uplc}
    A universe of elements $\B^n$ is considered.
    \begin{description}
        \item[Input] A circuit $P: \B^n \mapsto \B^{n-1}$.
        \item[Solutions] A set of $n$ distinct strings $a_1 \ldots a_n$ such that for every $j$ the strings $P(a_j), P(a_{j+1}) \ldots P(a_n)$ agree on the prefix of length $j$. 
    \end{description}
\end{definition}

We get the above definition of \UPLC by concatenating the circuits in the original definition \cite{Pasarkar2023}. We refer to the type of solution reported by \UPLC as the ``lower-triangular condition''. This is illustrated in \cref{fig:lower-triangular}.

\begin{figure}[t!]
    \centering
    \begin{subfigure}[t]{0.34\textwidth}
        \includegraphics[scale=0.3]{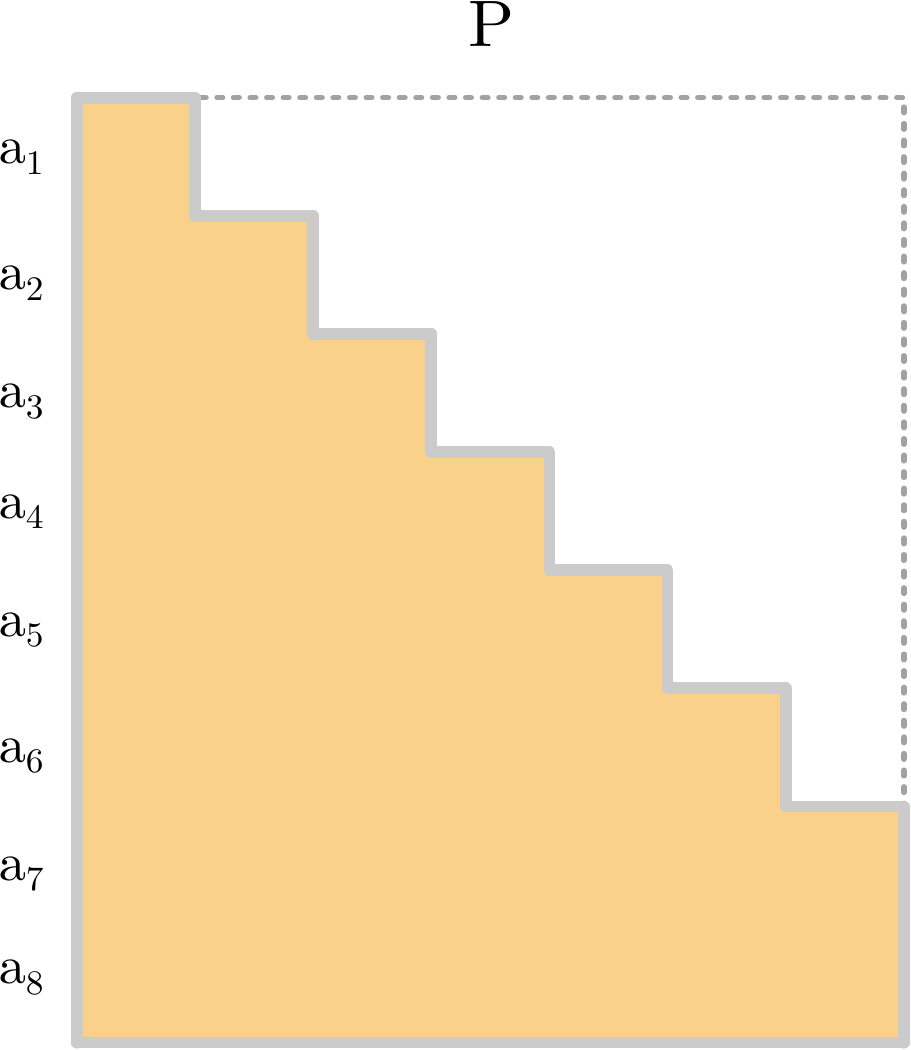}
        \caption{}
        \label{fig:lower-triangular}
    \end{subfigure}
    \begin{subfigure}[t]{0.3\textwidth}
        \includegraphics[scale=0.3]{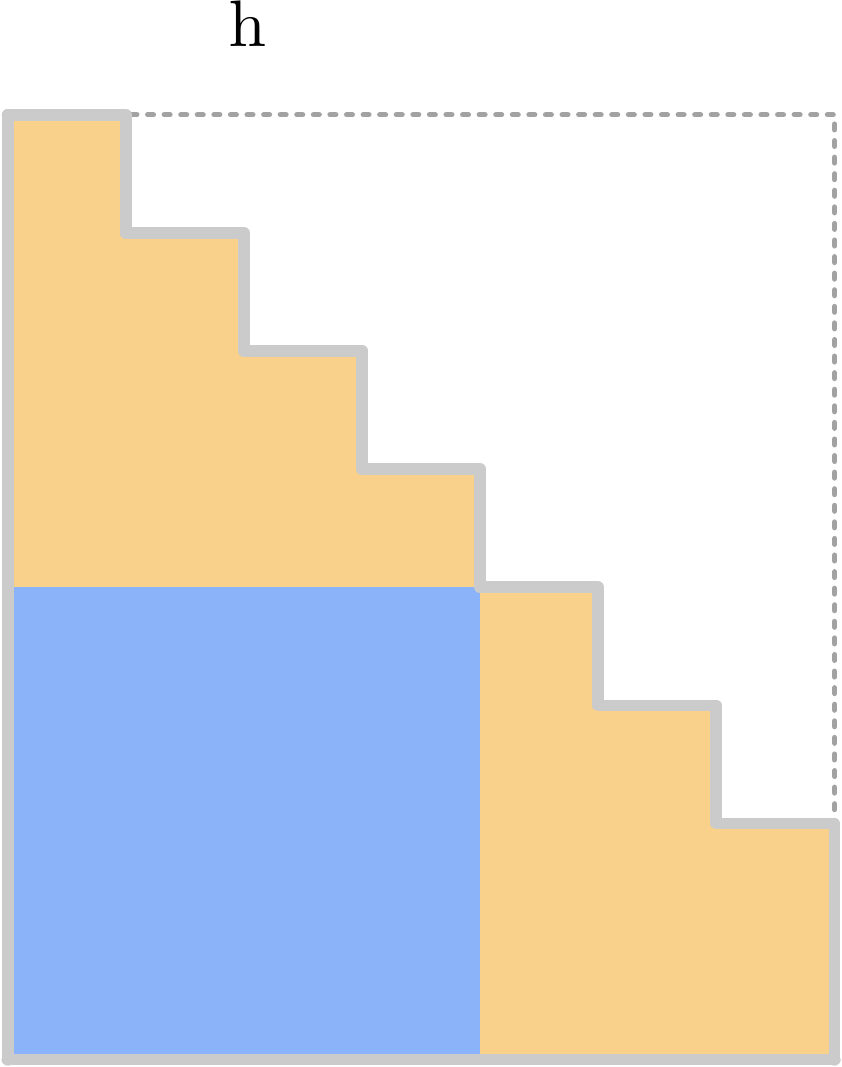}
        \caption{}
        \label{fig:pwpp-uplc}
    \end{subfigure}
    \begin{subfigure}[t]{0.3\textwidth}
        \includegraphics[scale=0.3]{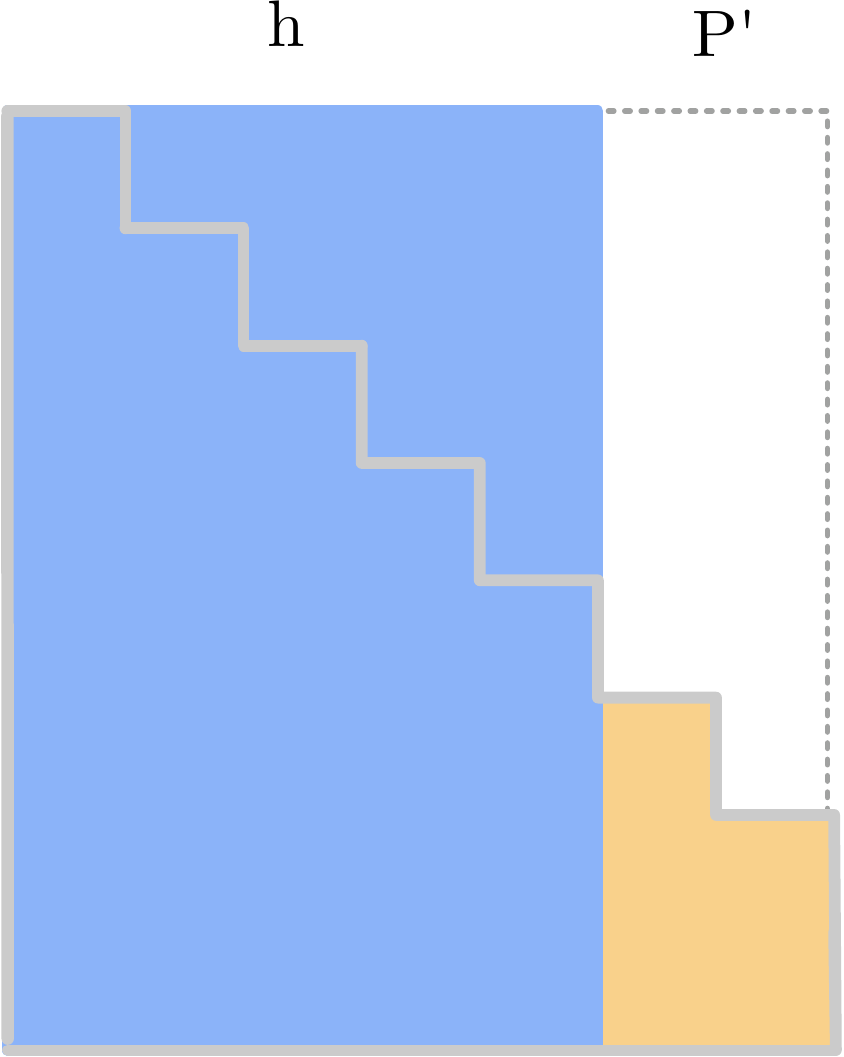}
        \caption{}
        \label{fig:uplc-pap}
    \end{subfigure}
    \caption{\UPLC solutions and reductions. \textbf{(a)} illustrates the lower-triangular condition of solutions to \UPLC. \textbf{(b)} is an illustration of \cref{lem:pwpp-uplc}. The \textcolor{newcolor}{blue} shaded region is a valid solution to $n/2$-\PWPP which is a subset of the \textcolor{YellowOrange}{yellow} shaded region, the lower-triangular collision returned by \UPLC. \textbf{(c)} is an illustration of \cref{lem:uplc-pap}. The \textcolor{newcolor}{blue} shaded region is the collision returned by $\tightPigeon{n}$, and the \textcolor{YellowOrange}{yellow} shaded region is computed by solving a $\log n$ size \UPLC instance by brute force. }
\end{figure}

\begin{lemma}
\label{lem:pwpp-uplc}
$\ourPigeon{n/2}{N}{\sqrt{N}}$ is in \UPLC.
\end{lemma}

\begin{proof}
Given an instance $(n, \pmap)$ of $\ourPigeon{n/2}{N}{\sqrt{N}}$, where $\pmap : \B^n \mapsto \B^{n/2}$, we construct an instance $P$ of \UPLC by simply appending 0s to the circuit.

$$P(x_1 x_2 \ldots x_{n-1}) = \pmap(x_1 \ldots x_{n/2}) 0^{n/2-1}.$$

The solution to \UPLC gives us the lower triangular condition for some set of strings $\{a_0, a_1 \ldots a_n\}$ which in particular gives us an $n/2$-collisions for $\pmap$, since $a_{\frac{n}{2}} \ldots a_{n}$ agree on the first $n/2$ bits. This proof is illustrated in \cref{fig:pwpp-uplc}.
\end{proof}

\uplcpap*
\begin{proof}
    This is by combining \cref{lem:pwpp-uplc} with \cref{thm:lb}.
\end{proof}

Further, we complete a fine-grained understanding of the position of \UPLC in the Pecking Order by placing it in $n$-\PWPP.

\begin{lemma}
\label{lem:uplc-pap}
    $\UPLC \subseteq n\textrm{-}\PWPP$.
\end{lemma}

\begin{proof}
    Without loss of generality, we assume $n$ is a power of 2. Given the circuit $P: \B^n \mapsto \B^{n-1}$ of a \uplc instance, we ``split'' the output into two parts: the first $n - \log n$ bits are considered as an instance of $\ourPigeon{n}{N}{N/n}$ and the remaining $\log n - 1$ bits are considered as an instance of \uplc in a $\log n$ scale.

    We first solve the $\ourPigeon{n}{N}{N/n}$ instance defined by a mapping $h: \B^n \mapsto \B^{n-\log n}$, where $h$ outputs the first $n - \log n$ bits of $P$. Note that $\ourPigeon{n}{N}{N/n} \in n\textrm{-}\PWPP$. Suppose we get an $n$-collision $(a_1, \ldots, a_n)$ of $h$ in the first step.
    
    We then consider the \uplc instance $P'$ on the universe $U = \{a_1, \ldots, a_n\}$, where $P': U \mapsto \B^{\log n - 1}$ is defined by the remaining $\log n - 1$ bits of $P$. We can solve this much smaller \uplc instance in polynomial time, and get a set of solution $(a_{i_1}, \ldots, a_{i_{\log n}})$. 
    
    Finally, we get the solution to the original \uplc instance by rearranging the elements in $U$: we put $(a_{i_1}, \ldots, a_{i_{\log n}})$ in the end with the same order, and put everything else in before with arbitrary order. The first $n - \log(n)$ columns of the lower-triangular condition are satisfied by finding the $n$-collision of $h$, and the remaining $\log n -1$ columns are fulfilled by solving the small \uplc instance $P'$.

    This proof is illustrated in \cref{fig:uplc-pap}.
\end{proof}

Combining \cref{thm:lb:compression}, \cref{lem:uplc-pap}, and the fact that $\PPP \subseteq \PLC$~\cite{Pasarkar2023}, we get two more black-box separations regarding to $\UPLC$.

\uplcsep*

\paragraph{A stronger version of \UPLC.}
We observe that \uplc is still total if the universe has size $2^{n-1}+1$, and we can even ask for more matching bits between those pigeons.
To see this, we apply the \emph{non-adaptive} version of the iterative PHP: Start with the universe $U_0 = [2^{n-1}+1]$. In the $i$-th step, the $i$-th bit of $P$ divides the current set $U_{i-1}$ into two subsets and sets $U_i$ to be the larger one (break ties arbitrarily). In this way, for any $i \in [n-1]$, all the elements in $U_i$ agree on the first $i$ bits of $P$, and $|U_{n-i}| \geq 2^i$.
Therefore, we can pick 
\begin{enumerate}
    \item two elements ($a_{n-1}$ and $a_n$) from $U_{n-1}$ that match on all $n-1$ bits of $P$,
    \item at least $2$ more elements from $U_{n-2}$ that match on first $n-2$ bits with $a_n$, \ldots,
    \item and in general, at least $2^{i-1}$ more elements from $U_{n-i}$ that match on first $n-i$ bits with $a_n$.
\end{enumerate}

We now define a new problem $\tuplc$ (``T'' for tight) to study the hardest computational problem corresponding to the non-adaptive iterative pigeonhole principle described above.

\begin{definition}[\tuplc]\label{def: tuplc}
    Let $n = 2^{\ell}$ be a power of $2$, and a universe of elements $U = [2^n+1]$ is considered.
    \begin{description}
        \item[Input] A circuit $P: U \mapsto \B^{n}$.
        \item[Solutions] A set of $n+1$ distinct strings $a_0 \ldots a_{n}$ such that for every $i \in \{0, 1, \ldots, \ell\}$, the strings $P(a_{n - {2^i}}), P(a_{n - {2^i} + 1}) \ldots P(a_n)$ agree on the prefix of length $n-i$. 
    \end{description}
\end{definition}

    \Xcomment{
        By definition, $\PPP$ and $\UPLC$ are the subclasses of $\TUPLC$.
        \begin{lemma}
            $\PPP \subseteq \TUPLC$, and $\UPLC \subseteq \TUPLC$.
        \end{lemma}
        \begin{proof}
            For $\PPP$, notice that $(a_{n},a_{n+1})$ is exactly the solution for the \pigeon instance $(n, \pmap)$ with $\pmap \coloneq P$.
        
            For $\UPLC$, let $n' = n-1$, and then an instance of $\TUPLC$ of size $n'$ is also an instance of $\uplc$ of size $n$, but with a smaller universe.
        \end{proof}
    }

We show the counter-intuitive result that the non-adaptive iterative pigeonhole principle is equivalent to the generalized pigeonhole principle computationally.
\begin{lemma}\label{lem: tuplc pap}
    $\tuplc$ is $\PAP$-complete.
\end{lemma}
\begin{proof}[Proof Sketch]
    Let $n = 2^{\ell}$.
    Note that the solution $\{a_0 \ldots a_{n}\}$ for a $\tuplc$ instance $(n, P)$ is also the solution for the $\ourPigeon{(n+1)}{N+1}{N/n}$ instance $(n, \pmap)$ with $\pmap$ being the first $n-\ell$ bits of $P$. Therefore, $\tuplc$ is $\PAP$-hard.
    
    To show that $\tuplc$ is also in $\PAP$, we use the same two-stage argument as in \cref{lem:uplc-pap}: 
    The first $n - \ell$ bits of $P$ are considered as an instance of $\ourPigeon{(n+1)}{N+1}{N/n}$ and we get an $(n+1)$-collision $U' \coloneqq \{a_1, \ldots, a_{n+1}\}$ in this step. The remaining $\ell$ bits of $P$ are interpreted as an instance of \tuplc in a $\log n$ scale with the universe $U'$, each can be solved trivially.
\end{proof}

\subsection{Total Function \BQP}\label{sec:quantum}

In this section, we put the problem defined by Yamakawa-Zhandry~\cite{Yamakawa2022} for their breakthrough result in the Pecking Order. Yamakawa-Zhandry's problem was first defined relative to a random oracle. It was later adapted to constitute a total problem~(\cite{Yamakawa2022}, Section 6), which is the version we considered in this paper.

\yzpap*

We start with formally defining the Yamakawa-Zhandry's problem.

\begin{definition}[Yamakawa-Zhandry's Problem~\cite{Yamakawa2022}, Simplified]\label{def:YZ}
    Fix an error correcting code $C \subseteq \Sigma^n$ on alphabet $\Sigma$. Let $(h_k)$ be a family of $\lambda$-wise independent functions from $C$ to $\{0,1\}^n$.

    \begin{description}
        \item[Input] The input encodes $n$ mapping $f_1, \ldots, f_n: \Sigma \rightarrow \{0, 1\}$. We define $f: \Sigma^n \rightarrow \B^n$ as $f(a_1 a_2 \ldots a_n) \coloneqq (f_1(a_1), f_2(a_2), \ldots, f_n(a_n))$. 
        \item[Solutions] The goal is to find a key $k$ and $t$ codewords $c^{(1)} \ldots, c^{(t)} \in C$ such that $f(c^{(i)}) \oplus h_k(c^{(i)}) = 0^n$, for all $i \in [t]$.
    \end{description}
    
    The parameters satisfy $ n < \lambda \ll t = \poly(n)$, $|C| \geq 2^{2n}$, and $|\Sigma| = 2^{\poly(n)}$.
\end{definition}

This definition is slightly different to the one appeared in \cite{Yamakawa2022}\footnote{In \cite{Yamakawa2022}, the input contains $t$ mappings $f_1, \ldots, f_t$, and the goal is find a key $k$ and $t$ codewords $c^{(1)} \ldots, c^{(t)} \in C$ such that $f_i(c^{(i)}) \oplus h_k(c^{(i)}) = 0^n, \forall i \in [t].$}, which is simpler to analyze. A \emph{folded Reed-Solomon} code (see \cite{guruswami23}, Section 17.1) with a certain parameter setting is chosen for the ECC $C$ in \cite{Yamakawa2022}. The structure of the ECC is crucial for the exponential quantum speed-up. It is straightforward to verify that the problem as stated here preserves the quantum upper bound and the classical lower bound. We assume the $\lambda$-wise independent functions family $(h_k)$ is implemented by the well-known low-degree polynomial construction.

\begin{definition}[\cite{Vadhan12psr}, Section 3.5.5]
\label{def:poly-k-wise}
    Let $\mathbb{F}$ be a finite field. Define the family of functions $\mathcal{H} = \{h_{a_0,a_1 \ldots a_{\lambda-1}} : \mathbb{F} \mapsto \mathbb{F}\}$, where each $h_{a_0,a_1 \ldots a_{\lambda-1}} = a_0 + a_1x + a_2x^2 \ldots a_{\lambda-1}x^{\lambda-1}$ for $a_0,a_1 \ldots a_{\lambda-1} \in \mathbb{F}$. 
\end{definition}

It is noted in \cite{Vadhan12psr} that this family forms a $\lambda$-wise independent set.

\paragraph{A further simplification.} To show \cref{lemma:yz}, we consider a simplification of the Yamakawa-Zhandry's Problem, in which we assume all the codewords have $n$ distinct letters, and no letter appears in two different codewords. We call this new problem \AZC. The rationale behind the naming is that we can now imagine we are given an \emph{arbitrary} 0/1 matrix $F$ of size $n \times |C|$ representing the mapping $f$, where the $i$-th column of $F$ corresponds to the result of applying $f$ to the $i$-th codeword.

\begin{definition}[\AZC]\label{def:azc} \hfill
    Fix a  0-1 matrix family $\mathcal{H}$ of size $n \times m$,  where the entries are $\lambda$-wise independent, implemented by the low-degree polynomial construction.

    \begin{description}
        \item[Input] The input is a 0-1 matrix $F$ of size $n \times m$.
        \item[Solutions] The goal is to find $t$ indices $j_1, \ldots, j_t \in [m]$ and a matrix $H_k \in \mathcal{H}$ (which could be succinctly represented by a key $k$), such that for any $i \in [t]$, the $j_i$-th column of matrix $F \oplus H_k$ is $\mathbf{0}^n$.
    \end{description}

    The parameters satisfy $ n < \lambda \ll t = \poly(n)$, and $2^n \cdot t \leq m = 2^{\poly(n)}$.
\end{definition}

As such, the Yamakawa-Zhandry's problem can be seen as \AZC\ with a \emph{promise}: $F$ but has the structure imposed by the ECC $C$ in the definition of the Yamakawa-Zhandry's problem. \cref{lemma:yz} is then implied by the following lemma.

\begin{lemma}\label{lem:azc_pap}
    \AZC\ is contained in \PAP.
\end{lemma}

\begin{proof}
    Given an \AZC\ instance specified by a $6$-tuple $(n,m, \lambda, t, F, \mathcal{H})$, we construct the following $\ourPigeon{t}{M}{N}$ instance for $M = m \cdot \frac{|\mathcal{H}|}{2^n}$ and $N = |\mathcal{H}|$.
    \begin{description}
        \item[Pigeons] Each pigeon is a pair $(j, k)$. Here $j$ is in $[m]$, and $k$ is a key for the family $\mathcal{H}$ such that the $j$-th column of matrix $F \oplus H_k$ is $\mathbf{0}^n$.
        \item[Holes] Each hole is a key $k$.
    \end{description}

    Pigeon $(j, k)$ is now mapped to hole $k$. Recall that we have $n < \lambda$, thus, each index $j$ could pair with exactly $2^{-n}$ fractions of keys to form a valid pigeon. Hence, there are $M = m \cdot \frac{|\mathcal{H}|}{2^n}$ pigeons. Since $m \geq 2^n \cdot t$ in the \AZC\ instance, we can verify that $M \geq t \cdot N$. Also, any $t$-collision in the $\ourPigeon{t}{M}{N}$ instance directly corresponds to a solution of the \AZC\ instance.

    Since we implemented $\mathcal{H}$ using the low-degree polynomial construction \cref{def:poly-k-wise}, this reduction can be implemented in $\poly(n)$ time in the white-box setting using polynomial interpolation.
\end{proof}

\Xcomment{
    \begin{lemma}[Totality]
        $\AZC_n$ is total whenever $f_k$ is a $n$-wise independent function ($\lambda\geq n$) and the set of words $W$ is large.
    \end{lemma}
    \begin{proof}
        Let $P:= $ number of words mapped to $0^n$ by the oracle XORed with our choice of $f_k$.
        Now, choosing $f_k$ randomly, we get that $\E[P] = \sum_{w} \mathbf{1}_{H\oplus f_k(w) = 0^n} = |W|/2^n$, where we used $n$-wise independence to be able to compute the summands. Whenever $|W| \geq t2^n$, there must exist a $f_k$ such that we have $t$ words mapped to $0^n$ since the max is at least average.
    \end{proof}
}

\subsection*{Acknowledgements}
We thank Beach House for providing topical procrastination music \cite{beachhouse}. We thank Mika G\"o\"os, Scott Aaronson, David Zuckerman, Alexandors Hollender, Gilbert Maystre, Ninad Rajgopal, Chetan Kamath for discussions about the Pigeonhole principle. We also thank Scott Aaronson for suggesting the name Pecking Order. SJ thanks Mika G\"o\"os for countless insightful discussions about \TFNP and complexity theory more generally.

SJ and JL are supported by Scott Aaronson’s
Vannevar Bush Fellowship from the US Department of Defense, the Berkeley NSF-QLCI CIQC
Center, a Simons Investigator Award, and the Simons “It from Qubit” collaboration. ZX is supported by NSF award CCF-2008868
and the NSF AI Institute for Foundations of Machine Learning (IFML).

\newpage
\DeclareUrlCommand{\Doi}{\urlstyle{sf}}
\renewcommand{\path}[1]{\small\Doi{#1}}
\renewcommand{\url}[1]{\href{#1}{\small\Doi{#1}}}
\bibliographystyle{alphaurl}
\bibliography{refs}

\appendix

\end{document}

%% file: tfnp.tex
\begin{tikzpicture}[scale=1.1]
\tikzset{inner sep=0,outer sep=3}


\tikzstyle{a}=[inner sep=6pt, inner ysep=6pt,outer sep=0.5pt,
draw=black!20!white, fill=Cerulean!2!white, very thick, rounded corners=6pt, align=center]

\begin{scope}[yscale=1.145]
\large
\node[a] (PLS) at (-3.5,3) {$\PLS$};
\node[a] (PPA) at (9, 3) {$\PPA$};
\node[a] (PPP) at (4,0) {$\PPP$};
\node[a] (3PPP) at (4,1.5) {3-\PPP};
\node[a] (4PPP) at (4,3) {4-\PPP};
\node[rotate=90] (Dots) at (4,4.5) {\ldots};
\node[a] (PiH) at (4,6) {$\PiH$};
\node[a] (SAP) at (4,7.5) {$\SAP$};
\node[a] (PAP) at (4,9) {$\PAP$};
\node[a] (UPLC) at (0,0) {$\UPLC$};
\node[a] (PLC) at (0, 6) {$\PLC$};
\node (BiRamsey) at (7, 3) {$\biramsey$};
\node (Ramsey) at (-1.5,3) {$\ramsey$};
\end{scope}

\path[-{Stealth[length=6pt]},line width=.6pt,gray]
(BiRamsey) edge (PAP)
(Ramsey) edge (PLC);

\path[-{Stealth[length=6pt]},line width=.6pt,YellowOrange]
(UPLC) edge (PAP)
(PPP) edge (PLC)
(PPP) edge (3PPP)
(3PPP) edge (4PPP)
(4PPP) edge (Dots)
(Dots) edge (PiH)
(PiH) edge (SAP)
(SAP) edge (PAP)
(UPLC) edge (PLC);

\small
\hypersetup{hidelinks}
\tikzset{new/.style={-{Stealth[length=6pt]},dashed,line width=1pt,YellowOrange}}

\draw[new,bend right=4]
(UPLC) edge
node[pos=0.64,sloped,above,inner sep=2pt]{\cref{cor:uplc-ppp}}
(SAP);

\draw[new]
(Ramsey) edge
node[pos=0.65,sloped,above,inner sep=2pt]{\cref{cor:ramsey-ppp}}
(SAP);

\draw[new]
(BiRamsey) edge
node[midway,sloped,below,inner sep=2pt]{\cref{cor:ramsey-ppp}}
(SAP);


\draw[new]
(PPA) edge
node[pos=0.75,sloped,above,inner sep=2pt]{\cref{thm:plspap}}
(PAP);

\draw[new]
(PPP) edge
node[midway,sloped,below,inner sep=2pt]{\cref{cor:uplc-plc}}
(UPLC);

\draw[new,bend left=14]
(PLS) edge
node[midway,sloped,above,inner sep=2pt]{\cref{thm:plspap}}
(PAP);

\normalsize
\tikzset{model/.style={color=MidnightBlue}}

\node [model, right=0 of PPP] {$=$ 2-\PPP};

\end{tikzpicture}